\documentclass[11pt]{article}

\usepackage{graphicx}
\usepackage{amssymb}

\usepackage{amsthm}
\usepackage{algorithm, algorithmic}

\newtheorem{thm}{Theorem}[section]
\theoremstyle{definition}

\theoremstyle{remark}

\theoremstyle{plain}
\newtheorem{lem}[thm]{Lemma}
\newtheorem{col}[thm]{Corollary}

\def\denseformat{
\setlength{\textheight}{9in}
\setlength{\textwidth}{6.9in}
\setlength{\evensidemargin}{-0.2in}
\setlength{\oddsidemargin}{-0.2in}
\setlength{\headsep}{10pt}
\setlength{\topmargin}{-0.3in}
\setlength{\columnsep}{0.375in}
\setlength{\itemsep}{0pt}
}

\denseformat
\begin{document}

%\begin{frontmatter}


\title{Deterministic Distributed $(\Delta + o(\Delta))$-Edge-Coloring, \\ and Vertex-Coloring of Graphs with Bounded Diversity}

\author{Leonid Barenboim\thanks{
\ Open University of Israel.
 E-mail: {\tt leonidb@openu.ac.il}
 \newline  $^{**}$  \ \  Ben-Gurion University of the Negev. Email: {\tt elkinm@cs.bgu.ac.il}
\newline $^{***}$ \ Open University of Israel. Email: {\tt tzali.tm@gmail.com} 
\newline This research has been supported by the Israeli Academy of Science, grant 724/15.
}  \and Michael Elkin$^{**}$ \and Tzalik Maimon$^{***}$}

%\address{TODO}
\begin{titlepage}
\def\thepage{}
\maketitle
\begin{abstract}
In the distributed message-passing setting a communication network is represented by a graph whose vertices represent processors that perform local computations and communicate over the edges of the graph. In the distributed {\em edge-coloring} problem the processors are required to assign colors to edges, such that all edges incident on the same vertex are assigned distinct colors. The previously-known deterministic algorithms for edge-coloring employed at least $(2\Delta - 1)$ colors, even though any graph admits an edge-coloring with $\Delta + 1$ colors \cite{V64}. Moreover, the previously-known deterministic algorithms that employed at most $O(\Delta)$ colors required superlogarithmic time \cite{B15,BE10,BE11,FHK15}. In the current paper we devise deterministic edge-coloring algorithms that employ only $\Delta + o(\Delta)$ colors, for a very wide family of graphs. Specifically, as long as the arboricity $a$ of the graph is $a = O(\Delta^{1 - \epsilon})$, for a constant $\epsilon > 0$, our algorithm computes such a coloring within {\em polylogarithmic} deterministic time.

We also devise significantly improved deterministic edge-coloring algorithms for {\em general graphs} for a very wide range of parameters. Specifically, for any value $\chi$ in the range $[4\Delta, 2^{o(\log \Delta)} \cdot \Delta]$, our $\chi$-edge-coloring algorithm has smaller running time than the best previously-known $\chi$-edge-coloring algorithms. Our algorithms are actually much more general, since edge-coloring is equivalent to {\em vertex-coloring of line graphs.}
Our method is applicable to vertex-coloring of the family of graphs with {\em bounded diversity} that contains line graphs, line graphs of hypergraphs, and many other graphs. We significantly improve upon previous vertex-coloring of such graphs, and as an implication also obtain the improved edge-coloring algorithms for general graphs. 

Our results are obtained using a novel technique that connects vertices or edges in a certain way that reduces clique size. The resulting structures, which we call {\em connectors}, can be colored more efficiently than the original graph. Moreover, the color classes constitute simpler subgraphs that can be colored even more efficiently using appropriate connectors. Hence, we recurse until we obtain sufficiently simple structures that are colored directly. We introduce several types of connectors that are useful for various scenarios. We believe that this technique is of independent interest.
\end{abstract}
\end{titlepage}
\pagenumbering {arabic}
%\begin{keyword}

%\end{keyword}

%\end{frontmatter}

\section{Introduction}
\label{Sc:intro}

\noindent {\bf 1.1 Distributed Synchronous Message Passing Model \\}
%\subsection{Distributed Synchronous Message Passing Model}
%\label{Sc:theModel}
In the synchronous message passing model a communication network is represented by an $n$-vertex graph $G = (V,E)$ of maximum degree $\Delta = \Delta(G)$. Each vertex in the graph has its own processing power and local memory. The vertices communicate via message passing over the edges of the graph. Computation proceeds in discrete rounds of local computations and exchange of messages. A single round is required for a message to traverse an edge.
The running time is the number of synchronized rounds until all vertices terminate. Generally, local computation is unrestricted. Vertices have distinct IDs of size $O(\log n)$.

%\subsection{Coloring Problems}
%\label{Sc:theProblems}
\noindent {\bf 1.2 Coloring Problems \\ }
The {\em vertex-coloring} and {\em edge-coloring} problems are among the most fundamental problems in distributed computing.
 These problems are significant in communication networks since appropriate labeling of stations, antennas and clients using a small number of colors is important for various network tasks. For example, channel allocation, scheduling, and work-load balancing can benefit considerably from a good coloring. In particular, edge-coloring is useful in open shop scheduling and scheduling production processes \cite{WHHHLSS97}, path coloring in fiber-optic communication \cite{EJ01}, and link-scheduling in sensor networks \cite{GDP05}.  The less colors one uses, the less cost one pays for performing a task. On the other hand, restricting the number of colors makes a problem more difficult.
 
 In 1987 Linial devised a deterministic algorithm for vertex-coloring and edge-coloring using quite a large number of colors, but with extremely good running time. Specifically, Linial obtained $O(\Delta^2)$-coloring that requires just $O(\log^* n)$ rounds. Roughly in the same time Cole and Vishkin \cite{CV86} and Goldberg et al. \cite{GPS87} devised deterministic algorithms for vertex-coloring of oriented paths and oriented trees with just $3$ colors. Since then these problems have been very intensively studied, resulting in a continuous progress both in general graphs and specific graph families. The most recent results make it possible to color vertices and edges of general graphs using $\Delta^{1 + \epsilon}$-colors in deterministic polylogarithmic time \cite{BE10,BE11}. However, when a smaller number of colors is used, the current results of distributed algorithms are far from being satisfactory. In particular, no deterministic polylogarithmic algorithms for $(\Delta + 1)$-vertex-coloring or $(2\Delta -1)$-edge-coloring are known. The situation with edge coloring is even worse, in the sense that although any graph can be edge-colored using at most $(\Delta + 1)$ colors \cite{V64}, no efficient deterministic $(\Delta + o(\Delta))$-edge-coloring algorithms is known in the distributed setting. (However, there are very efficient randomized algorithms for $\Delta(1 + \epsilon)$-edge-coloring, for any constant $\epsilon > 0$, that are applicable whenever $\Delta$ is sufficiently large \cite{DGP98,GP97,EPS15}.) 

In this paper we address the problem of edge-coloring with significantly less than $2\Delta - 1$ colors for a very wide family of graphs. Specifically, we obtain a deterministic algorithm that employs $\Delta + o(\Delta)$ colors for graphs with {\em arboricity}\footnote{The {\em arboricity} of a graph is the minimum number of forests into which its edge set can be decomposed. The arboricity is always smaller than or equal to the maximum degree of the graph.} $a = o(\Delta)$. Bounded arboricity graphs include planar graphs, graphs of bounded treewidth, graphs that exclude any fixed minor, and many other graphs. Unless the arboricity is very close to the maximum degree, the running time of our algorithm is {\em polylogarithmic} in $n$. %Hence we extend the knowledge regarding a central question of this setting: "what can be computed in deterministic polylogarithmic time?". Specifically, we compute an edge-coloring with ?-colors in ? time.

Our results are actually more general than that. For general graphs (with unbounded arboricity), for any positive integer $x \in o(\log \Delta)$ , we compute a $(2^{x+1}\Delta)$-edge-coloring within $\tilde{O}(x \cdot \Delta^{\frac{1}{2x + 2}} + \log^* n)$ time.\footnote{The $\tilde{O}$-notation supresses polylogarithmic factors. We note however, that whenever we write $\tilde{O}(f(\Delta) + \log^*n)$, the actual running time is $\tilde{O}(f(\Delta)) + O(\log^*n)$. We use the former notation for brevity.} This improves almost quadratically upon the state-of-the-art for the entire range of $x \in o(\log \Delta)$ (cf. Table 1). 
This table compares our new algorithms with the best previously-known ones. The latter are obtained by plugging in the state of the art $(\Delta + 1)$-vertex-coloring algorithm \cite{FHK15} in the algorithm of \cite{BE11}.
Our results address an open question raised in a PODC 2011 paper \cite{BE11} by the first-named and second-named authors of the current paper. % Even considering that in the current paper we use more recently developed auxiliary algorithms as building blocks, using the same auxiliary algorithm as in \cite{BE11} would still result in significant improvements in the whole range, thanks to the new powerful techniques that we derive in the current paper.
\begin{table}[h]

\begin{center}
  \begin{tabular}{|l|l|l|l||l|}
  \hline

  \multicolumn{2}{|c|}{{\bf Our Results}} & & \multicolumn{2}{|c|}{Previous Results (\cite{BE11} + \cite{FHK15})} \\
  \hline
  \hline
	Number of colors & Running time & & Number of colors & Running time \\
  \hline

  $4\Delta$ & $\tilde{O}(\Delta^{1/4}) + O(\log^* n)$ & & $(4 + \epsilon)\Delta$ & $O(\Delta^{1/3} + \log^* n)$ \\
  \hline
  $8\Delta$ & $\tilde{O}(\Delta^{1/6}) + O(\log^* n)$ & & $(8 + \epsilon)\Delta$ & $O(\Delta^{1/4} + \log^* n)$   \\
  \hline
  $16\Delta$ & $\tilde{O}(\Delta^{1/8}) + O(\log^* n)$ & & $(16 + \epsilon)\Delta$ & $O(\Delta^{1/5} + \log^* n)$   \\
  \hline
%  $(2^{x+1}\Delta)  $ & $\tilde{O}(x \cdot \Delta^{\frac{1}{2x + 2}}) + O(\log^* n)$ & & $(2^{x+1} + \epsilon)\Delta  $ & $O(x \cdot \Delta^{\frac{2}{2 + x}} + \log^* n)$ \\
% \hline
 $(2^{x+1}\Delta)  $  & $\tilde{O}(x \cdot \Delta^{\frac{1}{2x + 2}}) + O(\log^* n)$  & &  $(2^{x+1} + \epsilon)\Delta  $ &
  $\tilde{O}(x \cdot \Delta^{\frac{1}{x + 2}} + \log^* n)$ \\
%  & & & & (By plugging \cite{FHK15} into \cite{BE11}) \\
  \hline
  \end{tabular}
\end{center}
\caption{Edge coloring of general graphs.}
\end{table}

In addition to these results, we obtain even more general ones, that apply to {\em vertex-coloring} of a wide family of graphs. In particular, this family includes line-graphs. (Recall that an edge coloring of a graph is a vertex coloring of its line graph.) However, line graphs are just an example of a more general family of graphs, namely {\em graphs with bounded diversity}. We define {\em diversity} as the maximal number of maximal cliques a vertex belongs to.\footnote{A clique $Q$ is {\em maximal} in $G$ if there is no other clique in $G$ that strictly contains $Q$'s vertices. The number of cliques is counted with respect to a consistent clique identification whenever the family of graphs in hand allows this. Otherwise, each vertex identifies all maximal cliques it belongs to, which results in a consistent identification. In both cases, the cliques that a vertex belongs to contain all its neighbors. For example, in line graphs, each clique corresponds to a vertex of the original graph, and the clique vertices correspond to the edges incident on that vertex of the original graph.}  The diversity of a graph, which is denoted as $D = D(G)$, is the maximal diversity among all the vertices of the graph. Clearly, line graphs have diversity $2$. Indeed, each vertex of a line graph corresponds to an edge in the original graph. Each endpoint of this edge corresponds to a clique in the line graphs. Since an edge has two endpoints, the diversity of any line graph is at most $2$. More generally, for any positive constant $c$, the diversity of the line graph of a $c$-uniform hypergraph is $c$.  Given a graph of diversity $D$ and maximal clique size $S$, we obtain a $(D^{x+1}S)$-coloring within $\tilde{O}(x \cdot \sqrt{DS^{1/{(x+1)}}} + \log^* n)$ time, for any positive integer $x \in o(\log S)$. The best previous results that apply to this family of graphs are the results for graphs with {\em bounded neighborhood independence} \cite{BE11}. (The latter family generalizes graphs with bounded diversity, thus the results apply to these graphs as well.) Our new results compare favorably on this family of graphs. (Cf. Table 2.) Although not as general, the family of graphs with bounded diversity is still an important graph family that captures many characteristics of graphs with bounded neighborhood independence, but has some specific helpful properties that we use to obtain our improved results.

\begin{table}[h]

\begin{center}
  \begin{tabular}{|l|l|l|l||l|}
  \hline

  \multicolumn{2}{|c|}{{\bf Our Results}} & & \multicolumn{2}{|c|}{Previous Results (\cite{BE11} + \cite{FHK15})} \\
  \hline
  \hline
	Number of colors & Running time & & Number of colors & Running time \\
  \hline

  $D^2S$ & $\tilde{O}(\sqrt{D}S^{1/4}) + O(\log^* n)$ & & $(D^2 + \epsilon)\Delta$ & $O(D^2 \cdot \Delta^{1/3} + \log^* n)$ \\
  \hline
  $D^3S$ & $\tilde{O}(\sqrt{D}S^{1/6}) + O(\log^* n)$ & & $(D^3 + \epsilon)\Delta$ & $O(D^3 \cdot \Delta^{1/4} + \log^* n)$   \\
  \hline
  $D^4S$ & $\tilde{O}(\sqrt{D}S^{1/8}) + O(\log^* n)$ & & $(D^4 + \epsilon)\Delta$ & $O(D^4 \cdot \Delta^{1/5} + \log^* n)$   \\
  \hline
%  $D^{x+1}S $ & $\tilde{O}(x \cdot \sqrt{D S^{1/(x + 1)}}) + O(\log^* n)$ & & $(D^{x+1} + \epsilon)\Delta  $ & $O(xD^x \cdot \Delta^{\frac{2}{2 + x}} + \log^* n)$ \\
%  \hline
$D^{x+1}S $  & $\tilde{O}(x \cdot \sqrt{D} S^{1/(2x + 2)}) + O(\log^* n)$ & &
 $(D^{x+1} + \epsilon)\Delta  $ & $\tilde{O}(xD^x \cdot \Delta^{\frac{1}{x + 2}} + \log^* n)$ \\
 %& & & & (by \cite{FHK15} and \cite{BE11}) \\
 \hline
  \end{tabular}
\end{center}
\caption{Vertex-coloring of graphs with bounded diversity $D$ and maximal clique size $S \leq \Delta + 1$.}
\end{table}

%\subsection{Clique Decomposition as a New Approach}
%\label{Sc:theModel}
\noindent {\bf 1.3 Clique Decomposition as a New Approach \\}
In this paper we show a new approach for solving coloring problems.
In this approach we think of a clique as the "coloring worst enemy".\footnote{Note that in general there are graphs that contain no large clique, but still have large chromatic number. In particular, triangle-free graphs may have arbitrarily large chromatic number \cite{M55}. However, in the family of graphs with bounded diversity $D$ and maximal clique size $S$, the chromatic number is at most $D(S - 1) + 1$, and thus, graphs with smaller cliques are better for coloring.} Our goal is then breaking down these cliques by removing edges within them, thus breaking the cliques to smaller components of themselves. We refer to the new structure we obtain as a {\em connector}.
We then color the connector, and make sure that when we return the edges that were removed, the coloring remains proper. We then take this one step further by using the same concept recursively to obtain a better running time at the expense of more colors used, but still not too much colors.
This technique not only allows us to obtain better results than any previously known ones, but it is also easier to implement. We note that in general, identifying maximal cliques is an NP-hard problem. (Though, in the current distributed setting, a vertex can identify the maximal cliques it belongs to within a single round). However, in various graph families, such as line graphs (that are provided with their original graphs), a consistent identification of cliques is not hard, even sequentially.

Connectors turn out to be extremely useful in various scenarios. They allow us to compute {\em clique-decompositions}, in which the graph is partitioned into a bounded number of subgraphs in which the maximal cliques are significantly smaller than in the original graph. Together with the bound on diversity, this fact allows us to color these subgraph efficiently in parallel, and obtain a proper {\em vertex-coloring} for the entire input graph. Another type of connectors is very useful for {\em edge-coloring}. In this case it is used for computing a different kind of decompositions, in which subgraphs have {\em stars} of bounded size. These star-decompositions turn out to be very useful for edge-coloring. Yet another type of connectors allows us to compute a decomposition in which the subgraphs have both bounded arboricity and bounded degree. This allows us to bound the number of colors further, and obtain an ultimate edge-coloring with $\Delta + o(\Delta)$ colors of graphs with arboricity bounded away from $\Delta$.

%\subsection{Related Work}
%\label{sc:relatedwork}
\noindent {\bf 1.4 Related Work\\}
Panconesi and Rizzi \cite{PR01} devised a deterministic edge-coloring algorithm for general graphs with $(2\Delta - 1)$ colors that requires $O(\Delta + \log^* n)$ time. This result was recently improved in \cite{B15,FHK15} that obtained $(2\Delta - 1)$-edge-coloring within $\tilde{O}(\sqrt \Delta + \log^* n)$ time. Efficient deteministic edge-coloring algorithms that employ $O(\Delta)$ colors and $\Delta^{1 + \epsilon}$ colors were devised in \cite{BE11}. Czygrinow et al. devised a deterministic $O(\Delta \log n)$-edge-coloring algorithm with $O(\log^4 n)$ time \cite{CHK01}.
In parallel to our work, Ghaffari and Su obtained a $(2 + \epsilon)\Delta$-edge-coloring in deterministic time $O(\log^{11} n/\epsilon^3)$ \cite{GS16}.
Randomized edge-coloring 
%of various families of graphs
was studied in \cite{DGP98, EPS15, GP97}. For graphs with arboricity $a$, the algorithm of \cite{BE08} computes $(2\Delta - 1)$-edge-coloring within $O(a + \log n)$ time. For the family of graphs with $a = O(\Delta^{1 - \epsilon})$, for a constant $\epsilon > 0$, a $(\Delta + 1)$-vertex-coloring algorithm with deterministic polylogarithmic time was devised in \cite{BE10}. The latter result, however, does not imply edge-coloring algorithms in polylogarithmic time, since the arboricity of line graphs is $\Theta(\Delta)$. %Also, while $(\Delta + 1)$-vertex-coloring can be used for $(2\Delta - 1)$-edge-coloring, it is not clear if it can be used for edge coloring with $2\Delta - 2$ or less colors.
Moreover, if the algorithm of \cite{BE10} is used for edge-coloring via vertex-coloring of line graphs, the number of colors is at least $2\Delta - 1$.

Edge coloring  was also extensively studied in the PRAM model \cite{KS87,LSH96,FR96,HLS01,ZN98}.
There are known PRAM NC deterministic algorithms for $(\Delta+o(\Delta))$-edge-coloring \cite{BR91,MNN94},
obtained by derandomizing Karloff-Shmoys' randomized algorithm \cite{KS87}. 
Zhou and Nishizeki \cite{ZN98} also devised a deterministic PRAM NC $O(\mbox{max}\{{\Delta,a^2}\})$-edge-coloring of graphs with constant arboricity $a$.
However, these algorithms \cite{KS87,BR91,MNN94,ZN98} heavily rely  on assumptions of the PRAM model, i.e., that processors can efficiently fetch information concerning any edge or vertex in the graph. Hence these algorithms appear to be inapplicable to the distributed setting.

%\section{Preliminaries}
%\label{Sc:prem}
%A clique $Q$ in $G$ is {\em maximal}, if there is no other clique $T$ in $G$, for which $V(Q) \subset V(T)$.
%The {\em Diversity} of a vertex $v$ in a given graph is the number of distinct maximal cliques which $v$ belongs to.
%I think that in the current paper is is ok if the size is greater than 1.
%The diversity of a graph $G$ is the maximum diversity of its vertices. We denote the diversity of a graph by $D = D(G).$

%%  ================================================================== START

\section{Clique Decomposition}
\label{Sc:cd}

In this section we devise an algorithm that partitions an input graph into subgraphs in a way that reduces clique size. Specifically, each clique in each subgraph is smaller by a factor of $k > 1$ than a maximal clique in the original graph.
%%%@old-version: In this section we partition cliques in a manner in which we receive two subgraphs each of which  cliques are smaller by a factor of a given $k>1$.

For parameters $p,q > 0$, a {\em $(p,q)$-clique-decomposition} of a graph $G = (V,E)$ is a vertex partition $U_1,U_2,\dots,U_p$ of $V$, such that the maximum size of a clique in $G(U_i)$ is $q$, for all $i \in [p]$. 
A {\em clique master} is a unique vertex of a clique which has the highest ID in the clique. It is possible that a master of a clique belongs to another clique and may even be its master as well. Also, non-masters of certain cliques may be masters in other cliques they belong to.

%%%%%%%%%%%%%%%%%%5 figure 1
%\begin{figure}
%\includegraphics[width=\textwidth]{Picture3.png}
%\caption{A connector with $t = 4$ of a pair of cliques $Q,R$ that share a vertex $v$.}
%\end{figure}

We construct {\em Clique Decomposition} using the following structure, which we call a {\em connector}. A connector is constructed as follows.
Given a graph $G$ and an integer parameter $t > 1$, we start by choosing a master in each maximal clique of $G$. Clique identification and master selection is performed in $O(1)$ rounds in the distributed setting, since each clique has diameter $1$. Each such a master is responsible for the computation in its clique. Denote by ${\cal Q}$ the set of all identified maximal cliques in $G$.
Let us denote the size of a clique $Q$ by $S(Q)$. Each clique $Q$ of ${\cal Q}$ partitions its vertex set into subsets $V_1,V_2,\dots,V_k$ of size $t$ each. (Except for the last set $V_k$ that can be smaller.) Thus, we get $k= \left \lceil S(Q)/t \right \rceil$ subsets for each clique $Q$. 
%%%@old-version:In each  $Q_i$, we mark all edges $e \in E$ such that both endpoints of $e$ are in the same subset $V_j$ for some $1 \leq j \leq k$. Let us denote $G'$ the resulting subgraph.
%%%Let $G' = (V,E')$ denote the subgraph of $G$ such that the edge set $E' \subseteq E$ consists only of edges of $G(V_i$), where $i \in[k]$ and $V_i$ is a subset of vertices in a partition of some clique $Q$. In other words, for all cliques $Q$ of $G$, and all subsets $V_1,V_2,...,V_k$ in the partition of $Q$, an edge $e$ belongs to $E'$ iff both its endpoints are in $V_i$ for some $i \in [k]$. We will refer to $G'$ as a {\em Connector}.
Let $E' = \{ (u,v) \in E \ | \ u,v \in V_i(Q), \mbox{ for some clique } Q \in {\cal Q} \mbox{ and some } i \in [k] \}$ denote the set of edges connecting vertices from the same part $V_i$ of a maximal clique $Q \in {\cal Q}$.
%(for all maximal cliques and all parts $v_i$).
Let $G' = (V,E')$ be the graph with the original vertex set, and edge set $E'$. The graph $G'$ will be referred to as {\em connector}. (See Figure 1 in appendix A.)
Denote $D = D(G)$.

\begin{lem} \label{lem:cdeg}
The maximum degree $\Delta' = \Delta(G')$ of $G'$ is at most $D \cdot (t-1)$.
\end{lem}
\begin{proof}
By definition of diversity, for a vertex $v \in V$, the vertex $v$ belongs to at most $D$ maximal cliques. Each such a clique has size at most $t$. Thus, $v$ has at most $D \cdot (t-1)$ neighbors.
\end{proof}

%Let us denote the degree of $G'$ by $\Delta'$.
Now we can color $G'$ using the algorithm  \cite{FHK15}. We obtain a coloring $\varphi$ within running time $O(\sqrt {\Delta'} \log^{2.5}\Delta' +\log^* n)$. The numbers of colors of $\varphi$ is $\Delta'+1 \leq D(G)(t-1)+1$. We denote the right-hand-side by $\gamma$.

We consider the original graph $G$ with the above coloring $\varphi$ of the vertices of $G$. Let us denote by $G_i, 1 \leq i \leq \gamma$, the subgraphs induced by the subset of all vertices with the same color $i, 1 \leq i \leq \gamma$. Even though $\varphi$ is a proper coloring of $G'$, it is not necessarily a proper coloring of $G$. However, the following holds.

\begin{lem} \label{lem:deg}
The maximum degree $\Delta(G_i)$ of $G_i$ for each $1 \leq i \leq \gamma$, is at most $(k-1)D(G)$.
\end{lem} 
\begin{proof} For a vertex $v \in V$, the vertex $v$ belongs to at most $D$ maximal cliques. In each such clique $Q$, there could be at most $k-1$ neighbors with the same color as $v$. (Recall that the vertex set of $Q$ is partitioned into subsets $V_1,V_2,\dots,V_k$. Each such subset is a clique in $G'$, and therefore all the vertices of $V_j, 1 \leq j \leq k$, have distinct colors. Thus, $v$ has at most $k-1$ neighbors in $Q$ with the same color as that of $v$.) Hence, $v$ has at most $(k-1)D$ neighbors in $G$ with the same color as that of $v$.
\end{proof}

\begin{lem} \label{lem:divers}
For each vertex $v$ of $G_i$, $i \in [\gamma]$, it holds that:\\
(i) The size of all cliques of $G_i$ that $v$ belongs to is at most $k$.\\
(ii) The diversity of $v$ with respect to $G_i$ is at most $D = D(G)$, i.e., $D(G_i) \leq D(G)$.
\end{lem}
\def\APPt{
\begin{proof}
We begin by proving part (i) of the lemma. Denote by $Q$ a clique that $v$ belongs to in $G'$. By the definition of $G'$, $Q$ is a part of some clique $T$ of $G$. In the clique $T$ the vertex $v$ can have at most $k-1$ neighbors with the same $\varphi$ color. Thus, in $Q$ there could be at most $k-1$ neighbors for $v$.
%See the proof of Lemma \ref{lem:deg}.
Now we prove part (ii) of the lemma. Let us assume for contradiction that there is a vertex $v$ in $G_i$ with a diversity greater than $D$. Since each clique in $G_i$ is a subgraph of a clique in $G$, there must be two distinct maximal cliques $Q_1, Q_2$ (of $G_i$) which $v$ belongs to that are contained in the same maximal clique $Q$ in $G$. (This is because the number of maximal cliques in $G$ that contain $v$ is smaller than that in $G_i$.) Since $G_i$ is a vertex-induced subgraph of $G$, if two vertices with the color $i$ are connected by an edge in $G$, then they are connected by an edge in $G_i$. Let $u$ and $w$ be vertices of $Q_1$ and $Q_2$, respectively. Both $u$ and $w$ belong to $Q$.  Hence the edge $(u,w)$ belongs to $E(G)$, and thus also belongs to $E(G_i)$. Hence the set of vertices $V(Q_1) \cup V(Q_2)$ is fully connected, contradicting the maximality of $Q_1$ and $Q_2$. Thus the diversity of each $v \in G_i$ is at most $D$. By definition, the diversity of $G_i$ is the maximal diversity of its vertices, that is at most $D$.
\end{proof}
}

The proof of Lemma \ref{lem:divers} is found in Appendix B.
Recall that $\gamma = (t -1) \cdot D + 1$. By Lemmas \ref{lem:cdeg} and \ref{lem:divers}, for any integer $t > 1$, we can compute $((t -1) \cdot D + 1, k))$-clique-decomposition $V_1, V_2,..., V_{\gamma}$ of $G$. Recall that $k=\left \lceil S/t \right \rceil \leq S/t+1$.
In other words, for any integer $t > 1$, we obtain a $(t \cdot D, S/t+1)$-clique-decomposition. Now the maximal cliques in the subgraphs $G_1,G_2,...,G_{\gamma}$ are smaller than maximal cliques in $G$. Also, the diversities $D_1,D_2,...,D_{\gamma}$ of $G_1,G_2,...,G_{\gamma}$, respectively, are all at most $D$. We recursively apply our method to obtain yet smaller cliques at the expense of increasing the number of subgraphs. Once we obtain subgraphs with sufficiently small clique size we can color them directly.

Now we are ready to provide our coloring algorithm using clique-decomposition. (See Algorithm 1.) The algorithm accepts as input a graph $G$, the diversity $D = D(G)$, the size $S$ of the maximal clique of $G$, the parameter $t$ of the connector to be constructed, and the number of recursion levels $x$. The algorithm starts by computing a connector in line 1. Then the connector is colored with $\gamma$ colors in line 3. These color classes induce $\gamma$ subgraphs. The algorithm is invoked recursively in parallel on all these subgraphs in line 7. The recursion terminates in lines 9 - 13, where the subgraphs are colored directly. The colorings of the recursion levels are combined in line 15. Once the algorithm terminates, we have a proper coloring of the entire input graph. This completes the description of the algorithm.  See pseudo-code below.
\begin{algorithm}[H]
\caption{CD-Coloring($G, D, S, t, x$)}
\begin{algorithmic}[1]
\STATE $G'$ := compute the connector of $G$ with the parameter $t$. 
\STATE $\Delta' := D(t-1)$ /* $\Delta'$ is the maximum degree of $G'$ */
\STATE $\varphi$ := color $G'$ with $\gamma = \Delta'+1$ colors. /* using \cite{FHK15} */ 
\STATE Denote by $G_i, 1\leq i \leq \gamma$, the subgraphs induced by color classes of $\varphi$.
\IF {$x>1$}
\FOR {each $G_i$ in parallel}
\STATE $\psi$ := CD-Coloring($G_i, D, \left \lceil S/t \right \rceil, t, x-1$) \ \ \ \ \ \ \ /* cliques' size decreases by a factor of $t$ */
\ENDFOR
\ELSE 

\STATE /* $x = 1$ */
\FOR {each $G_i$ in parallel}
\STATE $\psi$ := color $G_i$ with $D \cdot (\left \lceil S/t \right \rceil -1) + 1$ colors. /* using \cite{FHK15} */ \\ /* Note that $D \cdot (\left \lceil S/t \right \rceil -1) + 1 \geq \Delta(G_i)+1$. So this number of colors is sufficient.  */.
\ENDFOR
\ENDIF
\RETURN $\langle \varphi, \psi \rangle$
\end{algorithmic}
\end{algorithm}

Denote by $S$ the size of a maximal clique in the original input graph $G$, and by $S_i$ the size of a maximal clique in a subgraph $G_i$, for all $i \in [\gamma]$. We recursively apply our method on $G_i, i \in [\gamma]$, in parallel. (See line 7 of Algorithm 1.) That is, we compute $(t \cdot D, S_i/t + 1)$-clique-decomposition in each $G_i$. In other words, this is a $(t \cdot D, (S/t)/t + 1/t + 1)$-clique-decomposition of each $G_i$, $1 \leq i \leq t \cdot D$. (Recall that the diversity of $G_i$ is not greater than the diversity of $G$.)
%we obtain  $(t \cdot D(G), (S_G/t)/t + 1/t + 1)$-clique-decomposition in each $G_i$. 
Thus, the overall number of subsets is $(t \cdot D)^2$.
Hence, after two recursion levels we obtain
$((t \cdot D)^2,  S/(t^2)+ 1/t  + 1)$-clique-decomposition of the original graph $G$.
After three recursion levels we obtain
$((t \cdot D)^3,  S/(t^3)+ 1/t^2 + 1/t + 1)$-clique-decomposition.
For $x$ recursion levels we obtain \\
$((t \cdot D)^x,  S/(t^x)+ 1/t^{x-1} + 1/t^{x - 2} + \dots + 1/t + 1)$-clique-decomposition.
Since $t \geq 2$, we have a geometric sequence $1/t^{x-1} + 1/t^{x - 2} + \dots + 1/t + 1 \leq 2$, and thus we obtain
$((t \cdot D)^x, S/t^x+2)$-clique-decomposition of $G$.
Since the maximum degree of each connector is $(t - 1) \cdot D$, by Lemma \ref{lem:cdeg}, the coloring of the connectors in each recursion level requires $\tilde{O}(\sqrt{D \cdot t} + \log^* n)$ time. Since there are $x$ recursion levels, we obtain the following Theorem.
\begin{thm}\label{decomRuntime}
We compute $((t \cdot D)^x, S/(t^x)+2)$-clique-decomposition of $G$
in time $\tilde{O}(x \cdot (\sqrt{D \cdot t} + \log^* n))$.
\end{thm}

In the end of the $x$ recursion levels, we obtain $(t \cdot D)^x$ subgraphs. We color these subgraphs as follows.
Let us denote the degree of a subgraph $G_i$, $1 \leq i \leq (t \cdot D(G))^x$, by $\Delta_i$. Now we use the algorithm of \cite{FHK15} again on each $G_i$ in parallel. (See line 12 of Algorithm 1.) We obtain a $(\Delta_i + 1)$-coloring within $\tilde{O}(\sqrt{\Delta_i} + \log^* n)  = \tilde{O}(\sqrt{D(G) S/(t^x)} + \log^* n)$ time. Let us denote the obtained coloring by $\psi$.
At the final stage, each vertex $v \in V$ colors itself with $\langle i , \psi (v) \rangle$, where $i$, $1 \leq i \leq (t \cdot D)^x$, is the index of the subgraph $G_i$ that $v$ belongs to (in the end of $x$ recursion levels). 

\begin{thm} The resulting coloring is proper.
\end{thm}
\begin{proof} Let $u,v$ be two adjacent vertices in $G$. If they belong to different subsets $G_i,G_j, i \neq j$, then $\langle i, \psi(u) \rangle \neq \langle j, \psi(v) \rangle$, and we are done. %If they both belong to some $V_j$ of some clique Q, then $\varphi(u) \neq \varphi(v)$.
Otherwise, there exists an index $j$, such that $v,u \in G_j$. Therefore, in the last stage of the algorithm, they obtain different colors $\psi$, and thus $\psi(u) \neq \psi(v)$. Hence, $\langle j, \psi (v) \rangle \neq \langle j, \psi (u) \rangle$.
\end{proof}

\begin{thm}\label{colCount} The resulting coloring employs $O((t \cdot D)^x \cdot (S/(t^x)+2) \cdot D)$ colors.
\end{thm}
\begin{proof} The number of subgraphs in the last stage is $(t \cdot D)^x$. The maximal clique size in each subgraph is $S/t^x+2$. Hence each vertex  has at most $S/t^x+1$ neighbors of the same clique. Thus, the maximum degree in each subgraph is at most $(S/(t^x)+1) \cdot D$. Therefore, the coloring $\psi$ (of the last stage)  employs at most $ (S/(t^x)+1) \cdot D + 1$ colors. The overall number of colors is $O((t \cdot D)^x \cdot (S/(t^x)+2) \cdot D)$.
\end{proof}

\begin{thm}\label{thm:nd} 
The running time of {\em CD-Coloring }is  $\tilde{O}(x \cdot (\sqrt{D \cdot t} + \log^* n) + \sqrt {D \cdot S/t^x} + \log^* n)$.
\end{thm}
\begin{proof} By Theorem \ref{decomRuntime}, computing the decomposition requires $\tilde{O}(x \cdot (\sqrt{D \cdot t} + \log^* n))$ time. The maximal clique size of the decomposition is $k \leq S/(t^x)+2$. This means that each vertex $v$ of $G_i$ has degree of at most $D \cdot (S/t^x+1)$. By \cite {FHK15}, the running time of the final stage of the algorithm is $\tilde{O}(\sqrt {D \cdot S/t^x} + \log^* n)$. %Thus, we obtain the running time in the theorem.
\end{proof}

\section{Choosing $x$ and $t$}
\label{Sc:chooseParamrs}
%@TODO: in this chapter the $O$-notation should be corrected to $\tilde{O}$
%\newline
Theorems \ref{colCount} and \ref{thm:nd} are optimized by selecting $t=S^{1/(x+1)}$. The $\tilde{O}(x \log^* n)$ factor can be eliminated, and only a single factor of $O(\log^* n)$ will remain. Specifically, we can compute a proper $O(\Delta^2)$-coloring once in the beginning using the algorithm of Linial \cite{L87}, and then employ these colors instead of IDs. Moreover, whenever distinct IDs are required for subgraphs of maximum degree $\Delta'$, we can compute $O(\Delta'^2)$-coloring of these subgraphs from an $O(\Delta^2)$-coloring of the input graph within $\log^* {\Delta} - \log^* \Delta'$ time.  (See \cite{BE09} for an analogous argument.) Consequently,  each additional $\log^*n$ term is replaced by a $\log^* {\Delta'} - \log^* \Delta''$ term, where $\Delta'$ is the degree of the previously computed subgraph, and $\Delta''$ is the degree of the currently computed subgraph (which is a subgraph of the former). Hence, the overall running time of these invocations is $O(\log^* n)$.

Note that for sufficiently large $S$, the expression $S/t^x = S/(S^{(x/x + 1)}) = S^{1/(x+1)}$ is super-constant for all $1 \leq x \leq o(\log S)$. In this range it holds that $S/t^x + 2 \leq (1 + 2t^x/S) \cdot S/t^x \leq (1 + \mu) S/t^x$, for an arbitrarily small constant $\mu > 0$. We note that actually $\mu = o(1)$. (If $S$ is a constant, we directly obtain a $(D(S-1)+1)$-coloring in $\tilde{O}(\sqrt{D} + \log^*n)$ time.) Hence, we obtain the following theorem.

\begin{thm} \label{colparamtr}
For all integers $1 \leq x \leq o(\log S)$, we compute $((1+\mu)S \cdot D^{x+1})$-coloring within $\tilde{O}(x \cdot \sqrt{D \cdot S^{1/(x+1)}} + \log^* n)$ running time.
\end{thm}

%We note that by a more accurate analysis we can obtain better results as follows.
We refine the analysis and improve upon Theorem \ref{colparamtr}. This is summarized below in Theorem \ref{colCountb}. Its proof is in Appendix B.
\begin{thm}\label{colCountb} We compute a $(D^{x+1} S)$-coloring, within  $\tilde{O}(x \cdot (\sqrt{D \cdot t} + \log^* n) + \sqrt {D \cdot S/t^x} + \log^* n)$ time.
\end{thm}

\def\APPa{
\begin{proof}
In the first stage, we choose $x = 1$ and $t = \left \lfloor \sqrt{S}\right \rfloor$. We construct the connector accordingly. Now each vertex has $D \cdot (t - 1)$ neighbors in the connector. We color the connector using \cite{FHK15} with $D \cdot (t - 1) + 1$ colors. Denote this coloring $\varphi$. Each color class of $\varphi$ induces a subgraph with cliques of maximum size $k = \left \lceil S /t \right \rceil$. Each vertex in each subgraph belongs to at most $D$ cliques, therefore the maximum degree in each subgraph is $D \cdot (k - 1)$. Now we color each subgraph using \cite{FHK15} again with $D \cdot (k - 1) + 1$ colors. Denote this coloring by $\psi$. The combination of the two colorings $\langle \varphi,\psi \rangle$ is a proper coloring of the entire graph. The overall number of colors is
\begin{eqnarray*}  & & D \cdot (t - 1) + 1) \cdot (D \cdot (k - 1) + 1)  \\ & = & D^2tk - Dt(D - 1) - Dk(D - 1) + (D-1)^2 \\ & \leq & D^2 S + D^2 \left \lfloor \sqrt S \right \rfloor - Dt(D - 1) - Dk(D - 1) + (D-1)^2 \\ & = & D^2 S + D^2 \left \lfloor \sqrt S \right \rfloor - D (D - 1)(t + k) + (D - 1)^2 \\ & \leq &
 D^2 S + D^2 \left \lfloor \sqrt S \right \rfloor - (D^2 -D )(2\left \lfloor \sqrt S \right \rfloor - 1) \\
 & \leq &
 D^2 S + D^2 \left \lfloor \sqrt S \right \rfloor - D^2 2\left \lfloor \sqrt S \right \rfloor + D^2 + 2D \left \lfloor \sqrt S \right \rfloor - D \\
 & = &
 D^2 S - D^2 \left \lfloor \sqrt S \right \rfloor + D^2 + 2D \left \lfloor \sqrt S \right \rfloor - D \\  & \leq & D^2 S + 2,
 \end{eqnarray*}
 
For $D \geq 2$ and $S$ larger than some constant. Since $DS$ is larger than the maximum degree of the graph, we can apply the basic reduction for $2$ rounds, and obtain $D^2S$-coloring. Recall that the basic color reduction computes a $(\Delta + 1)$-coloring from a $(\Delta + r)$-coloring within $r -1$ rounds, for any integer $r > 1$. This is achieved by iterating over the color classes $\Delta + 1, \Delta + r-1,..., \Delta + 2$, and for each color class, selecting in parallel proper colors from $[\Delta + 1]$.  %Using the reduction of Kuhn and Wattenhofer, we obtain $D^2S - D$ coloring within an additional round.
 
Denote the resulting algorithm ${\cal A}_1$. Next, we describe how to obtain an algorithm ${\cal A}_{i + 1}$ from ${\cal A}_i$ for $i = 1,2,\dots$, where ${\cal A}_i$ is an algorithm that computes a $(D^{i + 1} S)$-coloring. We choose $x = i + 1$,
 $t = \left \lfloor S^{1/(i + 1)} \right \rfloor$. Algorithm ${\cal A}_{i+1}$ starts by constructing a connector (with the parameter $t$) of the input graph $G$. Next, the algorithm computes a coloring $\hat{\varphi}$ of the connector using $D (t - 1) +1$ colors. 
Each subgraph induced by a color class of $\hat{\varphi}$ contains cliques of maximum size $k = \left \lceil S/t \right \rceil$.
Then the algorithm computes coloring $\hat{\psi}$ of subgraphs induced by color classes of $\hat{\varphi}$, using ${\cal A}_i$. The number of colors in each such coloring is $D^{i + 1} k$. %and within an additional round that employ Kuhn-wattenhofer education \cite{?}, we obtain number of colors  $D^{i + 1} k - D^i$.
The overall number of colors is at most

\begin{eqnarray*}
(D(t - 1) + 1) (D^{i + 1} k) & \leq &
(D(t - 1) + 1) (D^{i + 1} (S/t + 1)) \\ & = & 
D^{i + 2}S + D^{i + 2}t - D^{i + 2}(S/t) - D^{i + 2} + D^{i + 1}(S/t) + D^{i+1} \\
& \leq & D^{i + 2}S + D^{i + 2}t - (D-1)D^{i + 1}(S/t) \\ & \leq & D^{i + 2}S.
\end{eqnarray*}

The last inequality holds for any $S$ larger than some constant since $(D -1)(S/t) > Dt$, because $t = \left \lfloor S^{1/(i + 1)} \right \rfloor$ and $i \geq 2$.

Now we analyze the running time of the algorithm for some chosen number of iterations. Specifically, we analyze the running time of ${\cal A}_i$ for all $i=1,2,\dots$.
The running time of ${\cal A}_1$ is the time of coloring the connector and then coloring each color class in the resulting coloring. Coloring the connector requires $\tilde{O}(\sqrt {D(t - 1)} + \log^* n) = \tilde{O}(\sqrt {Dt} + \log^* n)$ rounds by \cite{FHK15}. Coloring the subgraphs induced by different color classes requires $\tilde{O}(\sqrt {D(\left \lceil S /t \right \rceil - 1)} + \log^* n) = \tilde{O}(\sqrt{D \cdot S/t} + \log^* n)$ time by \cite{FHK15}. Thus, ${\cal A}_1$ has running time $\tilde{O}(\sqrt {Dt} + \sqrt{D \cdot S/t} + \log^* n)$.

%Suppose that we are given a subgraph $G' \subseteq G$.
Next, we analyze the algorithm ${\cal A}_{i + 1}$ for coloring $G$. We assume inductively that ${\cal A}_i$ runs on any subgraph  $G'$ of $G$ within time $T_i = \tilde{O}(i \cdot (\sqrt{D \cdot t} + \log^* n) + \sqrt {D \cdot S'/t^i} + \log^* n)$, where $S'$ is the maximal clique size in $G'$. We prove that the running time of ${\cal A}_{i+1}$ executed on $G$ is $\tilde{O}(\sqrt{D \cdot t} + \log^* n) + T_i$. %\tilde{O}(i  \cdot (\sqrt{D \cdot t} + \log^* n) + \sqrt {D \cdot S_{G}/t^{i+1}} + \log^* n)$.
Denote by $S = S_G$  the maximal clique size in $G$.
In ${\cal A}_{i+1}$, it holds that $t = \left \lfloor S^{1/(i + 2)} \right \rfloor$. 
Algorithm ${\cal A}_{i+1}$ starts by using \cite{FHK15} to color the connector within time $\tilde{O}(\sqrt{Dt}+\log^* n)$.
The maximal clique size in the subgraphs induced by color classes of the resulting coloring is at most $k = \left \lceil S/t \right \rceil$. Hence, for each such subgraph $G'$, it holds that  $ S' \leq \left \lceil S/t \right \rceil$.
%For each such subgraph $G'$ we denote by $S_{G'}$ the size of the maximal clique in $G'$. 
Consequently, the algorithm ${\cal A}_i$ colors each subgraph inuced by a color class within time $T_i = \tilde{O}(i \cdot (\sqrt{D \cdot t} + \log^* n) + \sqrt {D \cdot S'/t^i} + \log^* n) = \tilde{O}(i \cdot (\sqrt{D \cdot t} + \log^* n) + \sqrt {D \cdot S/t^{i+1}} + \log^* n)$.
The overall running time of ${\cal A}_{i+1}$ executed on $G$ is $\tilde{O}(\sqrt{Dt}+\log^* n) + T_i = \tilde{O}((i + 1) \cdot (\sqrt{D \cdot t} + \log^* n) + \sqrt {D \cdot S/t^{i+1}} + \log^* n)$.
\end{proof}}
We note that for a special case of line-graphs it holds that $D(G)=2$, since each vertex in a line-graph (in which cliques are identified by vertices of the original graph) can belong to at most 2 maximal cliques.\footnote{The number of all maximal cliques that a vertex belongs to in a line-graph may be larger than $2$. However, if $G$ is the line graph of an original  graph $\bar{G}$, then each vertex of $\bar{G}$ corresponds to a clique in $G$. Since each vertex in $G$ is an edge with two endpoints in $\bar{G}$, the vertex belongs to two identified maximal cliques in $G$. Since the pairs of identified cliques that each vertex belongs to contain all its neighbors, the diversity is $2$.} The maximal clique size of a clique in a line graph of an input graph $G$ is at most $\max\{ \Delta= \Delta(G), 3\}$.
In addition, as we have explained, the term of $O(\log^* n)$ can appear exactly once in the running time, if we invoke only once the algorithm of Linial \cite{L87}, and then employ proper $O(\Delta^2)$-colorings instead of an $n$-coloring.  Hence, by setting $t = S^{1/{x + 1}}$ in Theorem \ref{colCountb}, we obtain the following result. 

\begin{thm} \label{colorimproved}
(i) For any positive integer $x \in o(\log S)$, we compute a $(D^{x+1}S)$-coloring of graphs with diversity $D$ within $\tilde{O}(x \cdot \sqrt{DS^{1/(x+1)}} + \log^* n)$ time. \\ 
(ii) For any positive integer $x \in o(\log \Delta)$ , we compute a $(2^{x+1}\Delta)$-edge-coloring within $\tilde{O}(x \cdot \Delta^{\frac{1}{2x + 2}} + \log^* n)$ time. 
\end{thm}

We now analyze the results for specific values of $x$.
For $x=1$ we obtain an $(S \cdot D^2)$-coloring within $\tilde{O}(\sqrt{D \cdot S^{\frac{1}{2}}} + \log^* n)$ time.
For $x=2$ we obtain an $(S \cdot D^3)$-coloring within $\tilde{O}(\sqrt{D \cdot S^{\frac{1}{3}}} + \log^* n)$ time.
For $x=3$ we obtain an $(S \cdot D^4)$-coloring within $\tilde{O}(\sqrt{D \cdot S^{\frac{1}{4}}} + \log^* n)$ time, and so forth.
%For $x=\log S$ we obtain a $((1 + \mu)S \cdot D(G)^{\log (S)+1})$-coloring within $O(\log S \cdot \sqrt{D(G) \cdot S^{1/(\log (S)+1)}} + \log^*(n))$ runtime.

%We now consider a special case of line-graphs. In this special case $D(G)=2$, since each vertex in a line-graph can belong to at most 2 maximal cliques. The maximal clique size of a clique in a line graph of an input graph $G$ is at most $\max\{ \Delta= \Delta(G), 3\}$. 
We now take a second look at the above results in the case of line graphs, assuming that $\Delta \geq 3$.
For $x=1$ we obtain a $4\Delta$-coloring within $\tilde{O}(\sqrt[4]{\Delta} + \log^* n)$ time.
For $x=2$ we obtain a $8\Delta$-coloring within $\tilde{O}(\sqrt[6]{\Delta} + \log^* n)$ time.
For $x=3$ we obtain a $16\Delta$-coloring within $\tilde{O}(\sqrt[8]{\Delta} + \log^* n)$ time, etc.
%For $x=logS$ we obtain a $(2(1 + \mu)S^2)$-coloring within $O(\log S \cdot S^{1/2(\log S+1)} + \log^*(n))$ runtime.

Another interesting corollary for line graphs is the following. Since we aim at poly-logarithmic time, we set $x=\log S/(\epsilon \log \log S)$, for some $\epsilon > 0$. By Theorem \ref{colorimproved}, this results in $(S \cdot 2^{\log S/(\epsilon \log \log S)+1}) = 2S^{1+1/(\epsilon \log \log S)}$ colors. The running time we achieve is
\begin{eqnarray*}
& & \tilde{O}(\log S/(\epsilon \log \log S) \cdot \sqrt{2 \cdot S^{1/((\log S/(\epsilon \log \log S))+1)}} + \log^* n)  % $$= O(\log S/(\epsilon \log \log S) \cdot \sqrt{2 \cdot 2^{\log S/((\log S/(\epsilon \log \log S))+1)}} + \log^*n) $$ $$\leq O(\log S/(\epsilon \log \log S) \cdot \sqrt{2 \cdot 2^{\log S/((\log S/(\epsilon \log \log S)))}} + \log^*(n))$$
%$$\leq \tilde{O}(\log S/(\epsilon \log \log S) \cdot \sqrt{2 \cdot 2^{\epsilon \log \log S}} + \log^* n)$$
%& = &\tilde{O}(\log S/(\epsilon \log \log S) \cdot \sqrt{\log^\epsilon(S)} + \log^* n)  
\leq  
\tilde{O}((\log S)^{1+\frac{\epsilon}{2}} + \log^* n).
\end{eqnarray*}

\section{Edge Coloring Using Star-Partition}
\label{Sc:sd}If we would use the above algorithm for coloring edges we would need to simulate the line-graph of the given graph and invoke on it the above algorithm. This simulation requires some local computation and resources. In this section we show a different technique that does not require this simulation. We do this by introducing a different type of connectors.
%\subsection{Star Partition}

Given a graph $G = V,E)$, an edge subset $\hat{E} \subseteq E$ is a {\em star} if there is vertex $u \in V$ that is shared by all edges in $\hat{E}$. The vertex $u$ is the {\em center} of the star.
The partition we present in this section decomposes the given graph $G$ into subgraphs with smaller {\em stars}. This time we partition the edge set rather than the vertex set.
For parameters $p,q > 0$, a {\em $(p,q)$-star-partition} of a graph $G = (V,E)$ is an edge partition $E_1,E_2,\dots,E_p$ of $E$, such that the maximum size of a star in $E_i$ is at most $q$, for all $i \in [p]$. 

We achieve this decomposition in the following way. For some integer $t>1$, each vertex $v$ in $G$ groups its edges into subsets, each of which is of size at most $t$. Each vertex $v \in V$ defines virtual vertices $v_1,v_2,...,v_k$, a vertex for each such subset. These vertices are simulated locally by $v$. We obtain $k = \left \lceil \Delta/t \right \rceil$ virtual vertices per vertex $v$. We define the {\em edge-connector} as $G' = (V', E')$ where $V'$ is the set of all virtual vertices and $E'$ consists of all edges in $E$ but an edge $e'$ of $E'$ replaces the original vertices of its endpoints with the corresponding vertices from $V'$. More precisely, each vertex $v$ enumerates the edges adjacent on it using the set $\{1,2,..,\Delta\}$, such that each edge of $v$ is assigned a distinct number. Now, each edge $(u,v)$ holds two numbers $l(u),l(v)$. For each such edge in $E$, the set $E'$ contains an edge $(u_i,v_j)$, where $i =\left \lceil l(u)/t \right \rceil$ and $j = \left \lceil l(v)/t \right \rceil$. (See Figure 2 in Appendix A.)
%\newline

%%%%%%%%%%%%%%%%%%%%%%%%%5figure
%\begin{figure}
%\includegraphics[width=\textwidth]{Picture2.png}
%\caption{Edge-connector with $t = 3$.}
%\end{figure}

Note that the maximum degree of the edge-connector is $t$. We employ the edge-connector as follows. In the first stage, we edge-color the connector using \cite{FHK15}. We obtain $(2t-1)$-edge-coloring $\varphi$ of the connector within $\tilde{O}(\sqrt t + \log^* n)$ time.

In the second stage we group edges from $G$ that have the same $\varphi$ color into subsets $\hat{E}_1,\hat{E}_2,\dots,\hat{E}_{2t-1}$. Since each vertex $v$ of $G$ has $k$ virtual vertices, and each such virtual vertex has all its edges colored with distinct colors, the number of edges of the same color adjacent on $v$ is at most $k$. Hence, the maximum degree of each $G(E_i)$, $i \in [2t-1]$, is $k = \left \lceil \Delta/t \right \rceil$.
Now we color these subgraphs in parallel using \cite{FHK15}. We obtain $(2k-1)$-edge-coloring $\psi$ in each subgraph within $\tilde{O}(\sqrt k + \log^* n)$ time.

Each edge has now a color which is a combination of two colors. The overall number of colors is $(2t-1) \cdot (2k-1)$. The overall running time of the two stages is $\tilde{O}(\sqrt t + \sqrt k + \log^* n)$.
We note that $(2t-1) \cdot (2k-1) = 4t \cdot \left \lceil \Delta/t \right \rceil - 2 \left \lceil \Delta/t \right \rceil - 2t + 1 \leq 4 \Delta + 2t - 2 \Delta /t +1$.
Let us choose $t = \left \lfloor \sqrt \Delta \right \rfloor$. We get $4 \Delta + 2t - 2 \Delta /t +1 \leq 4 \Delta +1$. Within an additional round the number of colors can be reduced to $4\Delta$. This gives us a $4 \Delta$-coloring algorithm with $\tilde{O}(\Delta^{\frac{1}{4}} + \log^* n)$ time.

As in the case of clique-decomposition we can compute star-partition recursively on each of the subgraphs. To this end, we set $t = \Delta^{1/(x + 1)} $, for a positive integer $x$. As usual, $x$ denotes the number of recursion levels. In each level we compute star-partition of the subgraphs. This increases the number of subgraphs and reduces the maximum size of stars in each subgraph. Once the size of all stars becomes sufficiently small, we color the subgraphs directly using \cite{FHK15}. This scheme is the same as that of Theorem \ref{colorimproved}, and thus we obtain the following result.

\begin{thm} \label{edgcol}
For $x = 1,2,...$, we compute a $(2^{x+1}\Delta)$-edge-coloring within $\tilde{O}(x \cdot \Delta^{\frac{1}{2x + 2}} + \log^* n)$ time. Moreover, this computation does not require maintaining the line graph by the original graph.
\end{thm}

%A star in a graph $G$ is a clique in the line-graph of $G$.
%Consequently, we can treat a clique as a star in all lemmas and theorems of sections \ref{Sc:cd} and \ref{Sc:chooseParamrs}. Hence, they still hold for the edge-coloring case, without the need to construct line graphs, but using star-decomposition and edge-connectors of the original graph.

\section{Edge-Coloring with $\Delta+o(\Delta)$ Colors of Bounded Arboricity Graphs}
An $H$-partition with degree $d$ of a graph $G = (V,E)$ is a partition of $V$ into $\ell$ subsets $H_1,H_2,...,H_{\ell}$, such that the number of neighbors of each $v \in H_i$, $i = 1,2,...,\ell$, in $\cup_{j=i}^{\ell} H_j$ is at most $d$. Given a graph $G$ with arboricity $a = a(G)$, an $H$-partition with degree $d = (2 + \epsilon)a$, for any constant $\epsilon > 0$, can be constructed in $O(\log n)$ time \cite{BE08}. Using this partition we obtain a $(\Delta + O(a))$-edge-coloring of $G$ within $O(a \log n)$ time. This is achieved using the following auxiliary algorithm.

%\subsection{Union of edge-colored graphs} \label{sc:unionec}
Suppose that we are given a graph $G = (V = A \cup B, E)$, in which $A \cap B = \emptyset$. 
Suppose that the maximum degree of $G(A)$ is at most $d$, and that for each vertex $v \in A$, the number of neighbors of $v$ in $A \cup B$ is also bounded by $d$. 
Moreover, the graph $G(A)$ already has a proper edge-coloring using $O(d)$ colors, and the graph $G(B)$ already has a proper edge-coloring using $\Delta + O(d)$ colors. Next we devise an algorithm that obtains a unified coloring of all the edges of $G$. The number of employed colors is $\Delta + O(d)$, and the running time is $O(d)$.

This algorithm proceeds as follows. Each vertex in $A$ labels its edges which cross to $B$ using unique labels from $\{1,2, \dots ,d\}$. Now we perform $d$ rounds. In each round $i, 1 \leq i \leq d$, all edges with label $i$ become active. (Recall that each such edge has endpoints in $A$ and $B$.) Denote by $E''_i$ such active edges of round $i$. For each edge $e \in E''_i$, its endpoint in $B$  collects the colors from all the edges incident to $e$. Then, the endpoint finds a new color available for $e$. Since $e$ has at most $d- 1 + \Delta -1$ neighboring edges, there must be an available color within a palette of size $\Delta + d - 1)$. Moreover, if a vertex $v$ in $B$ is shared by several crossing edges, it is still able to assign them colors in the same round. Indeed, the endpoints in $A$ of these crossing edges are not shared by multiple edges of  $E''_i$. (Otherwise, a vertex in $A$ would have assigned the same label to some of its edges. However, it has assigned unique labels.) Hence, all color assignments made by vertices in the same round are performed on sets of edges, such that no enpoint is shared by multiple sets. Therefore, any assignment of distinct available colors by $v$ to its edge set of label $i$ results in a proper coloring. 
Overall, there are $d$ such rounds until the algorithm terminates. Then, the entire graph $G$ is properly colored using $\Delta + O(d)$ colors. This discussion is summarized in the next lemma.

\begin{lem} \label{unionec}
Suppose that we are given a graph $G = (V = A \cup B, E)$, $A \cap B = \emptyset$, and the degrees of vertices of $A$ in $G$ are at most $d$. Moreover, the graph $G(A)$ already has a proper edge-coloring using $O(d)$ colors and the graph $G(B)$ already has a proper edge-coloring using $\Delta + O(d)$ colors. Then a proper edge-coloring of $G$ using $\Delta + O(d)$ colors can be computed in $O(d)$ time.
%Suppose that we are given a graph $G = (V = A \cup B, E)$, such that $A$ is an independent set, the degree of vertices of $A$ is $O(a)$, and the degree of vertices of $B$ is at most $\Delta$. Moreover, the graph $G(B)$ already has a proper edge-coloring using $\Delta + O(a)$ colors. Then a proper edge-coloring of $G$ using $\Delta + O(a)$ colors can be computed in $O(a)$ time.
%In particular, if $G$ is bipartite, we can color it with $\Delta + O(a)$ colors in $O(a)$ time.
\end{lem}

Now we prove the following theorem using this lemma.

\begin{thm} \label{arbcoloring}
A $(\Delta + O(a))$-edge-coloring of $G$ can be computed in $O(a \log n)$ time.
\end{thm}
\begin{proof}
Compute an $H$-partition of $G$ using \cite{BE08} in $O(\log n)$ time. Note that the maximum degree of each $H$-set is $O(a)$. Compute an $O(a)$-edge-coloring of the $H$-sets in parallel. Since each $H$-set has maximum degree $O(a)$, this is computed within $\tilde{O}(\sqrt{a} + \log^* n)$ using the algorithm presented in the current article in Section \ref{Sc:sd}. (Actually, this step can be computed much faster in the expense of increasing the constant of the number of colors $O(a)$. See Theorem \ref{edgcol}.)
Now we go over the sets from $H_{\ell}$ back to $H_1$. In each stage we color the edges that cross between $H_i$ and $H_{i +1} \cup ... \cup H_{\ell}$, using Lemma \ref{unionec}. (It holds that $A = H_i$, $B = H_{i+1} \cup ,..., \cup H_{\ell}$, $d = O(a)$.) This requires overall of $\Delta + O(a)$ colors. The running time of each stage is $O(a)$, and there are total of $O(\log n)$ such stages. This results in a $(\Delta + O(a))$-edge-coloring of $G$ within $O(a \log n)$ time, which proves the theorem.
\end{proof}
Next we demonstrate how to improve the running time using a yet another kind of connectors. To this end, we first need an acyclic orientation of the edges set of $G$ with bounded out-degree. As shown in \cite{BE08}, it is achieved by orienting edges that cross between different $H$-sets towards the sets with greater indices, and edges of the same $H$-set towards the endpoints with greater IDs. This results in an acyclic orientation with out-degree $d = (2 + \epsilon)a$. Now we construct the following {\em orientation-connector}. Each vertex $v \in V$ defines $k = \left \lceil \sqrt{\Delta} \right \rceil$ virtual vertices $v_1,v_2,...,v_k$. Then it groups its incoming edges into $k$ subsets of size at most $\Delta/k$ edges each. Each edge of the $i$th subset, $i \in [k]$, is connected to $v_i$ (and oriented towards $v_i$). The vertex $v$ also groups its outgoing edges into $\left \lceil \sqrt{d} \right \rceil$ subsets of size at most $\sqrt{d}$ edges each. Each edge of the $i$-th subset, $i \in \left \lceil \sqrt{d} \right \rceil$, is connected to $v_i$ (and oriented outwards of $v_i$).
See Figure 3 in Appendix A.

%%%%%%%%%%%%%%%%%%%%%%%%%%5figure
%\begin{figure}
%\includegraphics[width=\textwidth]{P1.png}
%\caption{Orientation-connector}
%\end{figure}

The above connector has arboricity at most $\sqrt{d}$, since the out-degree is bounded by $\sqrt{d} = O(\sqrt{a})$, and the orientation is acyclic. (See, e.g., \cite{BE08}.)
Moreover, its maximum degree is at most $\sqrt{\Delta} + O(\sqrt a)$. Therefore, using Theorem \ref{arbcoloring}, we obtain a $(\sqrt{\Delta} + O(\sqrt a))$-coloring $\varphi$ of the connector within $O(\sqrt a \log n)$ time. Let $\hat{E_i} \subseteq E$ be the set of edges colored by some color $i$ of $\varphi$. Then $G(\hat{E_i})$ has at most $k = \left \lceil \sqrt{\Delta} \right \rceil$ incoming edges and $O(\sqrt{d}) = O(\sqrt{a})$ outgoing edges. (This is because there are $k$ virtual vertices per original vertex, all incoming edges of a virtual vertex are of distinct colors, all outgoing edges of a virtual vertex are of distinct colors, and only $O(\sqrt{d}) = O(\sqrt{a})$ virtual vertices have outgoing edges.) Thus, the maximum degree of $G(\hat{E_i})$ is $\left \lceil \sqrt{\Delta} \right \rceil + O(\sqrt a)$ and its arboricity is $O(\sqrt{a})$. This allows us to color all subgraphs $\{G(\hat{E_i}) \ | \ i$ is a color of $\varphi \}$, in parallel using Theorem \ref{arbcoloring}. We obtain en edge-coloring $\psi$ in each subgraph that employs $\left \lceil \sqrt{\Delta} \right \rceil + O(\sqrt a)$ colors per subgraph. The running time of this step is  $O(\sqrt a \log n)$ as well. To summarize, we obtained a proper edge coloring $\langle \varphi, \psi \rangle$ of the entire input graph. The number of colors is $(\sqrt{\Delta} + O(\sqrt a)) \cdot (\left \lceil \sqrt{\Delta} \right \rceil + O(\sqrt a)) = \Delta + O(\sqrt{\Delta \cdot a}) + O(a)$. This directly implies the following.
\begin{thm}
For graphs with $a(G) = o(\Delta)$, we compute $(\Delta + o(\Delta))$-coloring within $O(\sqrt{a} \log n)$ time.
\end{thm}
%Next, we can employ this result for yet better coloring. Specifically, right after computing the initial orientation with out-degree $O(a)$, we group outgoing edges of each vertex into $O(a^{1/3})$ subsets of size at most $a^{2/3}$ each. Then the coloring $\varphi$ of the resulting orientation-connector employs $\sqrt{\Delta} + O(a^{2/3})$ colors. The running time for computing $\varphi$ is $O(a^{1/3} \log n)$. Next, we compute $\psi$-coloring of subgraphs. Each coloring employs $\left \lceil \sqrt{\Delta} \right \rceil + O(a^{1/3})$ colors, and the running time is  $O(a^{1/6} \log n)$. The overall number of colors is $(\sqrt{\Delta} + O(a^{2/3})) \cdot (\left \lceil \sqrt{\Delta} \right \rceil + O(a^{1/3})) = \Delta + O(\sqrt{\Delta} \cdot a^{2/3}) + O(a)$. For $a = o(\Delta^{3/4})$ this is $(\Delta + o(\Delta))$-coloring.

Next, we improve our result further. Let $q \geq (2 + \epsilon)$ be a positive parameter, for some arbitrarily small constant $\epsilon > 0$. Let $\hat{a} = q \cdot a$. We can compute an $H$-partition with out degree $\hat{a}$ in $O(\frac{\log n}{\log q})$ time \cite{BE08}.  Next we describe an algorithm for computing $(\Delta^{1/x} + \hat{a}^{1/x} + 3)^x$-edge-coloring in $O(a^{1/x} \cdot (x + \frac{\log n}{\log q}))$, time for any positive integer parameter $x$. The algorithm starts by defining  $\delta = \left \lfloor \Delta^{1 - 1/x} \right \rfloor$ virtual vertices $v_1,v_2,...,v_{\delta}$ and $\hat{\alpha} =  \left \lfloor \hat{a}^{1- 1/x} \right \rfloor$  virtual vertices $v'_1,v'_2,...,v'_{\hat{\alpha}}$ per each original vertex $v$ of $G$. Then, incoming edges of $v$ are grouped into subsets of size at most $\left \lceil \Delta^{1/x} + 1 \right \rceil$ each, and connected to $v'_1,v'_2,...,v'_{\delta}$, oriented towards the virtual vertices. In addition, outgoing edges of $v$ are grouped into subsets of size at most $\left \lceil \hat{a}^{1/x} + 1 \right \rceil$ each, and connected to $v'_1,v'_2,...,v'_{\hat{\alpha}}$, oriented outwards from the virtual vertices. Note that the resulting orientation-connector is a bipartite graph. Indeed, for any $v,u \in V(G)$, the vertices $v_i$ and $u'_j$ are not connected by an edge for any pair of indices $i \neq j$. Thus, all virtual vertices of the form $v_i$ are on one side in the bipartite graph, while all virtual vertices of the form $v'_i$ are on the other, for all indices $i$.
Note that all vertices on one side have degree at most $\left \lceil \Delta^{1/x} + 1 \right \rceil$, and all vertices on the other side have degree at most $\left \lceil \hat{a}^{1/x} + 1 \right \rceil$. Therefore, this orientation-connector can be colored with at most $\Delta^{1/x} + \hat{a}^{1/x} + 3$ colors within $O(\hat{a}^{1/x})$ time (See Theorem \ref{unionec}.) Denote this coloring by $\varphi$.

The edge-coloring $\varphi$ constitutes a partition of the original graph $G$ into $\Delta^{1/x} + \hat{a}^{1/x} + 3$ color classes. The subgraph of $G$ induced by each color class has maximum degree at most $\delta = \Delta^{1 - 1/x}$, and maximum out-degree at most $\hat{\alpha} = \hat{a}^{1- 1/x}$. By repeating the same procedure again in each subgraph in parallel, we obtain  $\Delta^{1/x} + \hat{a}^{1/x} + 3$ colors within each subgraph in time $O(\hat{a}^{1/x})$. The overall number of colors is now $(\Delta^{1/x} + \hat{a}^{1/x} + 3)^2$. The subgraph of $G$ induced by each color class has maximum degree at most $\Delta^{1 - 2/x}$, and maximum out-degree at most $\hat{a}^{1- 2/x}$. Overall, if we perform $x - 1$ such stages, we obtain a $(\Delta^{1/x} + \hat{a}^{1/x} + 3)^{x-1}$-coloring, such that each color class induces a subgraph in $G$ with  maximum degree at most $\Delta^{1/x}$, and maximum out-degree at most $\hat{a}^{1/x}$. All these subgraphs can be colored in parallel using $\Delta^{1/x} + \hat{a}^{1/x} + 3$ unique colors within $O(\hat{a}^{1/x} \log n)$ time, by Theorem \ref{arbcoloring}. This results in a proper edge coloring of the entire input graph. To summarize, we obtained the following result.

\begin{thm} \label{colorpower}
For any $q \geq (2 + \epsilon)$, $\hat{a} = q \cdot a$, and $ x = 1,2,...$, we obtain $(\Delta^{1/x} + \hat{a}^{1/x} + 3)^x$-coloring within $O(\hat{a}^{1/x} (x + \frac{\log n}{\log q}))$ time.
\end{thm}

Next, we analyze how large the arboricity of a graph can be to still allow for $\Delta + o(\Delta)$ colors in polylogarithmic time. 
%Theorem \ref{colorpower}.
%%%It holds that $$ (\Delta^{1/x} + \hat{a}^{1/x})^x = \sum_{k = 0}^x {x \choose k} \Delta^{\frac{x - k}{x}} \hat{a}^{\frac{k}{x}} \leq \Delta + x \Delta^{\frac{x-1}{x}}\hat{a}^{\frac{1}{x}} + x^2 \Delta^{\frac{x-2}{x}}\hat{a}^{\frac{2}{x}} + x^3 \Delta^{\frac{x-3}{x}}\hat{a}^{\frac{3}{x}} + ... + \hat{a}$$
%Let us bound $x$ by $O(\frac{\log \Delta}{\log \log \Delta})$.
%%%In order to have the desired number of colors, we need to satisfy $x^k \hat{a}^{\frac{k}{x}} = o(\frac{\Delta^{k/x}}{x})$. Therefore, we need $\hat{a} = o(\frac{\Delta}{x^{x (k + 1)/k}})$. Since $k$ is at most $x$, we need $\hat{a} = o(\frac{\Delta}{x^{x+1}})$. Let us bound $x$ by $\epsilon \cdot \frac{\log \Delta}{\log \log \Delta}$, for an arbitrarilly small constant $\epsilon > 0$. Then, as long as $\hat{a} = o(\frac{\Delta}{x^{x + 1}}) = o(\Delta^{1 - \epsilon})$ the number of colors is $\Delta + o(\Delta)$. (We note that actually, $\hat{a}$ should be bounded by $\mbox{min} (\epsilon \cdot \frac{\log \Delta}{\log \log \Delta}, \log \hat{a})$, since selecting larger number of levels $x$ does not improve running time.) To summarize this discussion, by setting $\hat{a} = (2 + \epsilon)a$ and $x = \epsilon \cdot \frac{\log a}{\log \log a}$ we obtain the next theorem.
%%%\begin{thm}
%%%The family of graphs with arboricity $o(\Delta^{1 - \epsilon})$ can be edge-colored with $\Delta + o(\Delta)$ colors within $O(\log a \log n)$ time.  
%%%\end{thm}
We set $x = \frac{\log \hat{a}}{c \log \log \hat{a}}$, for a (possibly large) constant $c > 0$. Hence $\hat{a}^{1/x} = \log^c \hat{a}$. Let $\eta > 0$ be a parameter. Suppose that we have
\begin{equation}\label{eq1}
\Delta^{1/x} \geq \frac{x}{\eta}(\hat{a}^{1/x} + 3).
\end{equation}
Then the number of colors is $(\Delta^{1/x} + \hat{a}^{1/x} + 3)^x = \Delta(1 + \frac{\hat{a}^{1/x} + 3}{\Delta^{1/x}})^x \leq \Delta(1 + \eta/x)^x \leq \Delta \cdot e^{\eta} \leq \Delta (1 + 2\eta)$, for a sufficiently small $\eta > 0$. The condition (\ref{eq1}) means that
\begin{equation}\label{eq2}
\Delta \geq \frac{x^x}{\eta^x}(\hat{a}^{1/x} + 3)^x
\end{equation}
Set $q = 2 + \epsilon'$, for some constant $\epsilon' > 0$. (Recall that $\hat{a} = q \cdot a$.) Then $\Delta \geq a^{1 + 1/c + \frac{\log (2/ \eta)} {c \log \log a}}$ implies (\ref{eq2}) and (\ref{eq1}). If $a > \Delta^{\frac{1}{4\log \log \Delta}}$, then we set $\eta = O(1/\log \Delta)$, and obtain $\Delta(1 + O(\frac{1}{\log \Delta}))$-edge-coloring, assuming $\Delta \geq a^{1 + O(1/c)}$. The running time is $O(x \cdot a^{1/x} + a^{1/x} \log n) = O(\log^c a \log n)$.

If $a < \Delta^{\frac{1}{4\log \log \Delta}}$, then we set $x = \log \hat{a} = \log (a q)$. Here we do not set $q = 2 + \epsilon'$, but rather use a larger value of this parameter. Specifically, we set $q = \frac{1}{a} \cdot 2^{\frac{\log \Delta}{\log \log \Delta + \log \frac{1}{\eta}+ 1}}$, and $\eta = \frac{1}{\log \Delta}$. The running time becomes $O(\log \hat{a} + \frac{\log n}{\log q}) = O(\log a + \log q + \frac{\log n}{\log q}) = O(\frac{\log \Delta}{\log \log \Delta} + \frac{\log n}{\log \Delta/ (2 \log \log \Delta) - \log a})$.
This expression is always $O(\log n)$, and it is $o(\log n)$ whenever $\Delta = \omega(1)$. The number of colors is again $\Delta(1 + O(\frac{1}{\log \Delta}))$.
\begin{col}
Whenever $a < \Delta^{\frac{1}{4\log \log \Delta}}$, a variant of our algorithm computes a $\Delta(1 + O(\frac{1}{\log \Delta})) = \Delta(1 + o(1))$-edge-coloring in $O(\log n)$ time. Moreover, the time is $o(\log n)$ whenever $\Delta = \omega(1)$. \\
For larger arboricity, for any arbitrarily large positive parameter $c > 0$, a variant of our algorithm computes $\Delta(1 + O(\frac{1}{\log \Delta}))$-edge-coloring in time $O(\log^c a \log n)$, assuming $\Delta \geq a^{1 + O(1/c)}$.
\end{col}
In other words, the larger the gap between the maximum degree $\Delta$ and the arboricity $a$, the faster is our algorithm. Moreover, the algorithm runs in polylogarithmic time whenever the gap is at least polynomial, i.e., $\Delta \geq a^{1 + \epsilon}$, for some constant $\epsilon > 0$. \\
%%%%%%%%% add a result anout $(a^{epsilon} log n)$ time.

\noindent{\bf Acknowledgements\\}
\noindent The authors are grateful to Alessandro Panconesi for fruitful discussions.

%%  ================================================================== END

\section{Bibliography}
\label{Sc:bib}

\clearpage
\pagenumbering{roman}
\appendix
\centerline{\LARGE\bf Appendix}
\section{Figures}

%\begin{figure}
\includegraphics[width=\textwidth]{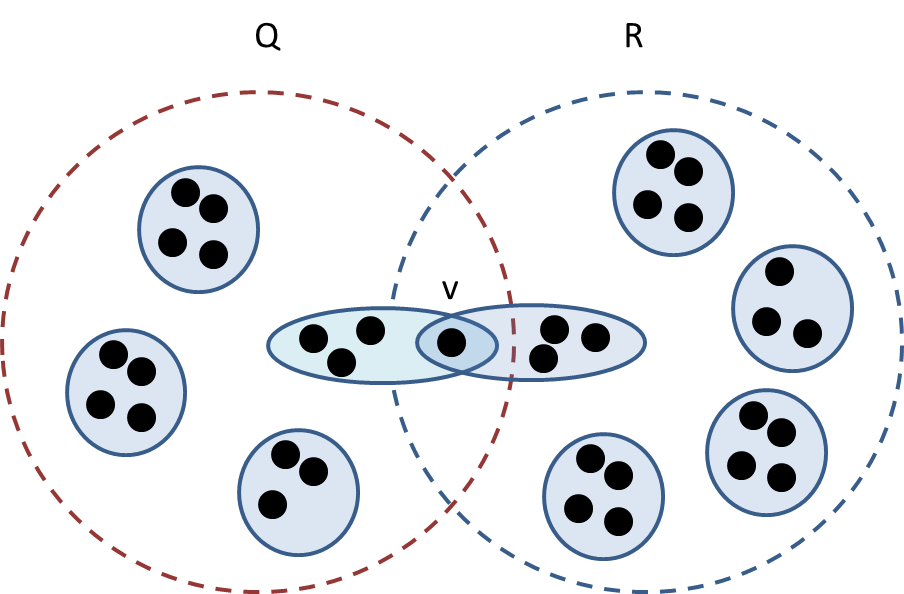}
%\caption{A connector with $t = 4$ of a pair of cliques $Q,R$ that share a vertex $v$.}
%\end{figure}
\\ \\
\centerline{ Figure 1.  A connector with $t = 4$ of a pair of cliques $Q,R$ that share a vertex $v$. }
\clearpage

\includegraphics[width=\textwidth]{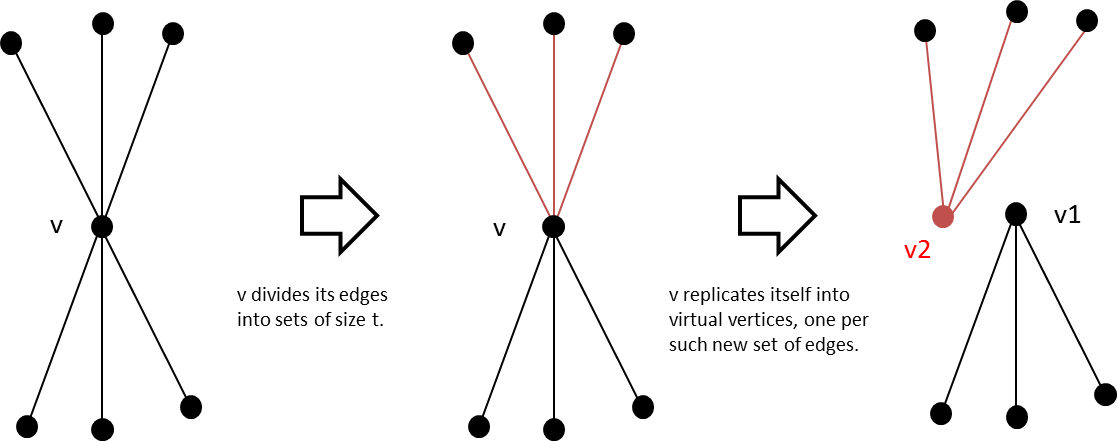}
\\ \\
\centerline{ Figure 2. Edge-connector with $t = 3$.} \\ \\ \\ \\ \\ \\ \\
\includegraphics[width=\textwidth]{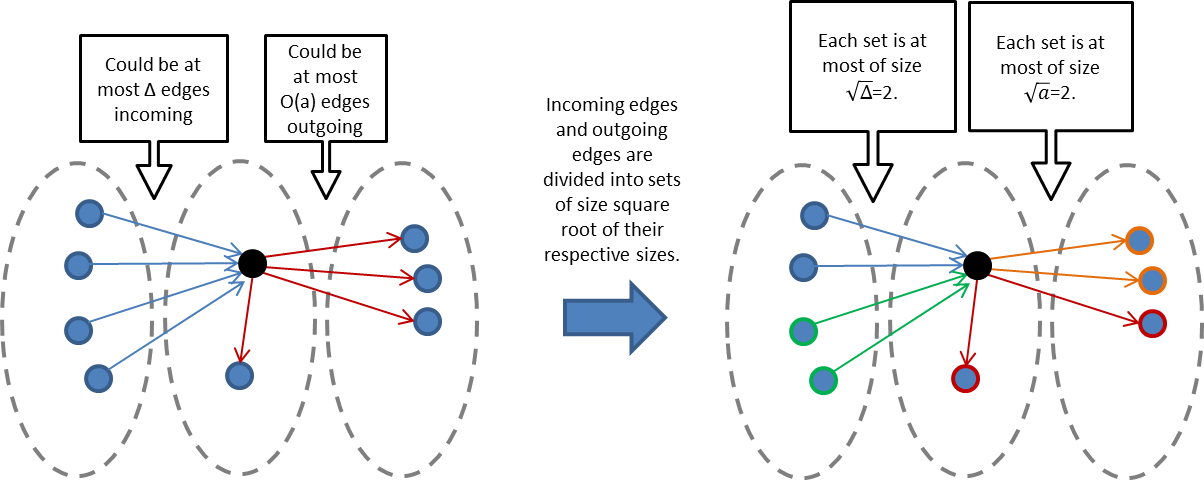} \\ \\
\centerline{ Figure 3. Orientation connector.}
\clearpage

\section{Proofs}
\noindent \textbf{Proof of Lemma \ref{lem:divers}} \APPt 

\noindent \textbf{Proof of Lemma \ref{colCountb}} \APPa 


\begin{thebibliography}{10}
%\begin{thebibliography}{10}
%\bibliography{sample.bib}

%\bibitem{ABCP96}
%B.~Awerbuch, B.~Berger, L.~Cowen, and D.~Peleg.
%\newblock  Fast Distributed Network Decompositions and Covers.
%\newblock {\em Journal of Parallel and Distributed Computing}, 39(2):105-114, 1996.

%\bibitem{AKS80}
%M.~Ajtai, J.~Komlos, and E.~Szemeredi.
%\newblock A note on Ramsey numbers.
%\newblock {\em Journal of Combinatorial Theory, Series A}, 29:354-360, 1980.



\bibitem{ABI86}
N.~Alon, L.~Babai, and A.~Itai.
\newblock A fast and simple randomized parallel algorithm for the maximal independent set problem.
\newblock {\em Journal of Algorithms}, 7(4):567--583, 1986.

%\bibitem{AS08}
%N.~Alon, and J.~Spencer.
%\newblock {\em The probabilistic method.}
%\newblock Wiley, 3rd ed., 2008.

%\bibitem{AH77}
%K.~Appel, and W.~Haken.
%\newblock The solution of the four color map problem.
%\newblock {\em Scientific American}, 237(4): 108-121, 1977.

%\bibitem{AW04}
%H.~Attiya, and J.~Welch.
%\newblock {\em Distributed Computing: Fundamentals, Simulations, and Advanced Topics.}
%\newblock Wiley, 2nd ed., 2004.

\bibitem{AGLP89}
B.~Awerbuch, A.~V. Goldberg, M.~Luby, and S.~Plotkin.
\newblock Network decomposition and locality in distributed computation.
\newblock In {\em Proc. of the 30th Annual Symposium on Foundations of Computer Science}, pages 364--369, 1989.

%\bibitem{AP90}
%B.~Awerbuch, and D,~Peleg.
%\newblock Sparse partitions.
%\newblock In {\em Proc. of the 31st IEEE Symp. on Foundations of Computer Science.} pages  503–513, 1990.

%\bibitem{B12}
%L.~Barenboim.
%\newblock On the locality of some NP-complete problems.
%\newblock In {\em Proc. of the 39th International Colloquium on Automata, Languages, and Programming}, part II, pages 403-415, 2012.
\bibitem{B15}
L.~Barenboim.
\newblock Deterministic $(\Delta + 1)$-Coloring in Sublinear (in $\Delta$) Time in Static, Dynamic and Faulty Networks.
\newblock In {\em Proc of the 34th ACM Symp. on Principles of Distributed Computing}, pages 345-354, 2015.


\bibitem{BE08}
L.~Barenboim, and M.~Elkin.
\newblock Sublogarithmic distributed MIS algorithm for sparse graphs using Nash-Williams decomposition.
\newblock In {\em Proc. of the 27th ACM Symp. on Principles of Distributed Computing}, pages 25--34, 2008.

\bibitem{BE09}
L.~Barenboim, and M.~Elkin.
\newblock Distributed $({\Delta} + 1)$-coloring in linear (in ${\Delta}$) time.
\newblock In {\em Proc. of the 41th ACM Symp. on Theory of Computing}, pages 111-120, 2009.

\bibitem{BE10}
L.~Barenboim, and M.~Elkin.
\newblock Deterministic distributed vertex coloring in polylogarithmic time.
\newblock In {\em Proc. 29th ACM Symp. on Principles of Distributed Computing}, pages 410-419,  2010.

\bibitem{BE11}
L.~Barenboim, and M.~Elkin.
\newblock Distributed deterministic edge coloring using bounded neighborhood independence.
\newblock In {\em Proc. of the 30th ACM Symp. on Principles of Distributed Computing}, pages 129 - 138, 2011.

\bibitem{BE13}
L.~Barenboim, and M.~Elkin.
\newblock {\em Distributed Graph Coloring: Fundamentals and Recent Developments.}
\newblock Morgan-Claypool Synthesis Lectures on Distributed Computing Theory, 2013.

\bibitem{BEK14}
L.~Barenboim, M.~ Elkin, and F.~Kuhn.
\newblock Distributed (Delta+1)-Coloring in Linear (in Delta) Time.
\newblock {\em SIAM Journal on Computing}, 43(1): 72-95, 2014.
\bibitem{BEPS12}
L.~Barenboim, M.~Elkin, S.~Pettie, and J.~Schneider.
\newblock The locality of distributed symmetry breaking.
\newblock In {\em Proc. of the 53rd Annual Symposium on Foundations of Computer Science}, pages 321-330, 2012.

%\bibitem{BGS98}
%M.~Bellare, O.~Goldreich, and M.~Sudan.
%\newblock Free bits, PCPs, and nonapproximability - towards tight results. 
%\newblock {\em SIAM Journal on Computing}, 27(3):804--915, 1998.

\bibitem{BR91}
B.~Berger, and J.~Rompel. 
\newblock Simulating $(\log cn)$-wise independence in NC.
\newblock {\em Journal of the ACM}, 38(4), pages 1026-1046. 1991.

%\bibitem{B98}
%B.~Bollobas.
%\newblock {\em Modern Graph Theory.}
%\newblock Springer, corrected edition, 1998.

%\bibitem{BBMRY11}
%P.~Berman, A.~Bhattacharyya, K.~Makarychev, S.~Raskhodnikova, and G.~Yaroslavtsev.
%\newblock Improved approximation for the directed spanner problem. 
%\newblock In {\em Proc. of the 38th International Colloquium on Automata, Languageas, and Programming}, pages 1-12, 2011.

%\bibitem{BKP14}
%T.~Bisht, K.~Kothapalli , S.~Pemmaraju.
%\newblock Super-fast t-ruling sets (Brief Announcement).
%\newblock In {\em Proc. of the 33th ACM Symposium on Principles of Distributed Computing}, pages 379-381,  2014.

%\bibitem{B04}
%B.~Bollobas.
%\newblock {\em Extremal Graph Theory.}
%\newblock Dover Publications, 2004.

%\bibitem{B01}
%B.~Bollobas.
%\newblock {\em Random Graphs.}
%\newblock Cambridge University Press, 2nd ed., 2001.

%\bibitem{BDRRS12}
%C.~Busch, C.~Dutta, J.~Radhakrishnan, R.~Rajaraman, and S.~Srinivasagopalan.
%\newblock Split and join: strong partitions and universal Steiner trees for graphs.
%\newblock in {\em Proc. of the 53rd Annual IEEE Symposium on Foundations of Computer Science}, pages 81 - 90, 2012. 


%\bibitem{CMWZZ94}
%B.~Chen, M.~Matsumoto, J.~Wang, Z.~Zhang, and J.~Zhang.
%\newblock A short proof of Nash-Williams' theorem for the arboricity of a graph.
%\newblock {\em Graphs and Combinatorics}, 10(1): 27-28, 1994.

%\bibitem{CK08}
%R.~Cole, and L.~Kowalik. 
%\newblock New linear-time algorithms for edge-coloring planar graphs. 
%\newblock {\em Algorithmica}, 50(3): 351-368, 2008.

\bibitem{CV86}
R.~Cole, and U.~Vishkin.
\newblock Deterministic coin tossing with applications to optimal parallel list ranking.
\newblock {\em Information and Control}, 70(1):32--53, 1986.

%\bibitem{C93}
%L.~Cowen.
%\newblock On Local Representations of Graphs and Networks. {\em Ph.D. Thesis, MIT.} 1993.

%\bibitem{CCW86}
%L.~Cowen, R.~Cowen, and D.~Woodall.
%\newblock Defective colorings of graphs in surfaces: partitions into subgraphs of bounded valence. 
%\newblock {\em Journal of Graph Theory}, 10:187--195, 1986.

%\bibitem{CGJ97}
%L.~Cowen, W.~Goddard, and C.~Jesurum.  
%\newblock Coloring with defect
%\newblock In {\em Proc. of the 8th ACM-SIAM Symp. on Discrete Algorithms,} pages 548--557, 1997. 


\bibitem{CHK01}
A.~Czygrinow, M.~Hanckowiak, and M.~Karonski.
\newblock Distributed O(Delta logn)-edge-coloring algorithm.
\newblock In {\em Proc. of the 9th Annual European Symposium on Algorithms}, pages 345--355, 2001. 

%\bibitem{D54}
%B.~Descartes.
%\newblock Solution to advanced problem No. 4526.
%\newblock {\em American Mathematical Monthly}, 61, page 532, 1954.

%\bibitem{D10}
%R.~Diestel.
%\newblock {\em Graph Theory.}
%\newblock Springer, 4th ed., 2010.

%\bibitem{DGPV08}
%B.~Derbel, C.~Gavoille, D.~Peleg, and L.~Viennot.
%\newblock On the locality of distributed sparse spanner construction.
%\newblock In {\em Proc. of the 27th ACM Symp. on Principles of distributed Computing}, pages 273-282, 2008.

%\bibitem{DK11}
%M.~Dinitz, and R.~Krauthgamer.
%\newblock Directed spanners via flow-based linear programs.
%\newblock In {\em Proc. of the 43rd ACM Symp. on Theory of Computing}, pages 323-332, 2011.

\bibitem{DGP98}
D.~Dubhashi, D.~Grable, and A.~Panconesi.
\newblock Nearly-optimal distributed edge-colouring via the nibble method.
\newblock {\em Theoretical Computer Science, a special issue for the best papers of ESA95}, 203(2):225--251, 1998.

%\bibitem{DMPRS05}
%D.~Dubhashi, A.~Mei, A.~Panconesi, J.~Radhakrishnan, and A.~Srinivasan.
%\newblock Fast distributed algorithms for (weakly) connected dominating sets and linear-size skeletons.
%\newblock {\em Journal of Computer and System Sciences}, 71(4):467-479, 2005.

%\bibitem{DP09}
%D.~Dubhashi and A.~Panconesi.
%\newblock {\em Concentration of Measure for the Analysis of Randomized Algorithms}.
%\newblock Cambridge University Press, 2009.

\bibitem{EJ01}
T.~Erlebach, and K.~Jansen.
\newblock The complexity of path coloring and call scheduling.
\newblock {\em Theoretical Computer Science}, 255 (1–2), pages 33–50, 2001.

%\bibitem{E59}
%P.~Erd\H{o}s.
%\newblock Graph theory and probability.
%\newblock {\em Canadian Journal of Mathematics}, 11: 34-38, 1959.

%\bibitem{E07}
%M.~Elkin.
%\newblock A near-optimal distributed fully dynamic algorithm for maintaining sparse spanners.
%\newblock In {\em Proc. of the 26th ACM Symp. on Principles of Distributed Computing}, page 185-194, 2007.

%\bibitem{EP01}
%M.~Elkin, and D.~Peleg.
%\newblock The client-server 2-spanner problem with applications to network design. 
%\newblock In {\em Proc. of the 8th International Colloquium on Structural Information and Communication Complexity}, pages 117-132, 2001.




%\bibitem{EP05}
%M.~Elkin, and D.~Peleg.
%\newblock Approximating $k$-spanner problems for $k \ge 2$.
%\newblock {\em Theoretical Computer Science}, 337(1-3): 249-277, 2005.

\bibitem{EPS15}
M.~Elkin, S.~Pettie, and H.~Su.
\newblock $(2\Delta - 1)$-Edge-Coloring is Much Easier than Maximal Matching in the Distributed Setting.
\newblock In {\em Proc. of the 26th ACM-SIAM Symp. on Discrete Algorithms}, pages 355-370, 2015.


%\bibitem{EFF85}
%P.~Erd\H{o}s, P.~Frankl, and Z.~F\"uredi.
%\newblock Families of finite sets in which no set is covered by the union of $r$ others.
%\newblock {\em Israel Journal of Mathematics}, 51:79--89, 1985.

%\bibitem{FK98}
%U.~Feige, and J.~Kilian.
%\newblock Zero Knowledge and the chromatic number.
%\newblock {\em Journal of Computer and System Sciences} 57:187--199, 1998.

\bibitem{FHK15}
P.~Fraigniaud, M.~Heinrich, and A.~Kosowski.
\newblock Local Conflict Coloring.
\newblock http://arxiv.org/abs/1511.01287.
\newblock To appear in FOCS 2016.

\bibitem{FR96}
M.~F\"urer, and B.~Raghavachari.
\newblock Parallel edge coloring approximation.
\newblock {\em Parallel processing letters}, 6(3), pages 321-329, 1996.

\bibitem{GDP05}
S.~Gandham, M.~Dawande, and R.~Prakash.
\newblock Link scheduling in sensor networks: distributed edge coloring revisited.
\newblock In {\em Proc. of the IEEE 24th Annual Joint Conference of the IEEE Computer and Communications Societies}, pages 2492–-2501, 2005.

%\bibitem{G68}
%T.~Gallai.
%\newblock On directed graphs and circuits.
%\newblock {\em Theory of Graphs (Proceedings of the Colloquium Tihany 1966), New York: Academic Press}, pages 115-118, 1968.

%\bibitem{GJ76}
%M.~Garey, and D.~Johnson.
%\newblock The complexity of near-optimal graph coloring.
%\newblock {\em Journal of ACM}, 23(1): 43-49, 1976.

%\bibitem{GV07}
%B.~Gfeller, and E.~Vicari.
%\newblock A randomized distributed algorithm for the maximal independent set problem in growth-bounded graphs.
%\newblock In {\em Proc. of the 26th ACM Symp. on Principles of Distributed Computing}, pages 53-60, 2007.

\bibitem{GS16}
M.~Ghaffari and H.~Su.
\newblock Distributed Degree Splitting, Edge Coloring, and Orientations.
\newblock https://arxiv.org/abs/1608.03220.
\newblock To appear in SODA 2016.

\bibitem{GPS87}
A.~Goldberg, S.~Plotkin, and  G.~Shannon.
\newblock Parallel symmetry-breaking in sparse graphs.
\newblock {\em SIAM Journal on Discrete Mathematics}, 1(4):434--446, 1988.

%\bibitem{GHS14}
%M.~G\"o\"os, J.~Hirvonen, and J.~Suomela.
%\newblock Linear-in-delta lower bounds in the LOCAL model.
%\newblock In {\em proc. of the 33th ACM Symp. on Principles of Distributed Computing}, pages 86-95, 2014.

\bibitem{GP97}
D.~Grable, and A.~Panconesi.
\newblock Nearly optimal distributed edge colouring in O(log log n) rounds.
\newblock {\em Random Structures and Algorithms}, 10(3): 385-405, 1997.

%\bibitem{GP98}
%D.~Grable, and A.~Panconesi.
%\newblock Fast distributed algorithms for Brooks-Vizing colourings.
%\newblock In {\em Proc. of the 9th Annual ACM-SIAM Symposium on Discrete Algorithms}, pages 473-480, 1998.

%\bibitem{G59}
%H.~Grotzsch.
%\newblock Zur Theorie der diskreten Gebilde, VII: Ein Dreifarbensatz fur dreikreisfreie Netze auf der Kugel.
%\newblock {\em Wiss. Z. Martin-Luther-U., Halle-Wittenberg, Math.-Nat. Reihe} 8: 109-120, 1959.

\bibitem{HLS01}
Y.~Han, W.~Liang, and X.~Shen.
\newblock Very fast parallel algorithms for approximate edge coloring.
\newblock {\em Discrete applied mathematics}, 108(3), 227-238, 2001.

\bibitem{HKP01}
M.~Hanckowiak, M.~Karonski, and A.~Panconesi.
\newblock On the distributed complexity of computing maximal matchings.
\newblock {\em SIAM Journal on Discrete Mathematics}, 15(1):41--57, 2001.

%\bibitem{HJ85}
%F.~Harary, and K.~Jones.
%\newblock Conditional colorability II: Bipartite variations.
%\newblock {\em Congressus Numer}, 50:205-218, 1985.

%\bibitem{HW80}
%G.~Hardy, and E.~Wright.
%\newblock {\em An introduction to the theory of numbers.}
%\newblock Oxford university press, 5th edition, 1980.

%\bibitem{H96}
%J.~Hastad.
%\newblock Clique is Hard to Approximate Within $n^{1-\epsilon}$. 
%\newblock In {\em Proc. of the 37th Annual Symposium on Foundations of Computer Science}, pages 627-636, 1996.

%\bibitem{HS12}
%J.~Hirvonen, and J.~Suomela.
%\newblock Distributed maximal matching: greedy is optimal.
%\newblock In {\em Proc. of the 31th ACM Symp on Principles of Distributed Computing}, pages 165-174, 2012.

%\bibitem{II86}
%A.~Israeli, and A.~Itai.
%\newblock A fast and simple randomized parallel algorithm for maximal matching.
%\newblock {\em Information Processing Letters}, 22(2):77-80, 1986.

%\bibitem{IS86}
%A.~Israeli, and Y.~Shiloach.
%\newblock An Improved Parallel Algorithm for Maximal Matching.
%\newblock {\em Information Processing Letters}, 22(2):57-60, 1986.

%\bibitem{JRS01}
%L.~Jia, R.~Rajaraman, and R.~Suel.
%\newblock An efficient distributed algorithm for constructing small dominating sets.
%\newblock In {\em Proc. of the 20th ACM Symp. on Principles of Distributed Computing}, pages 33-42, 2001.

%\bibitem{JLR00}
%S.~Janson, T.~Luczak, and A.~Rucinski.
%\newblock {\em Random Graphs.}
%\newblock Wiley-Interscience, 2000.

%\bibitem{J99}
%\"{O}.~Johansson.
%\newblock Simple distributed $(\Delta + 1)$-coloring of graphs.
%\newblock {\em Information Processing Letters}, 70(5):229--232, 1999.

\bibitem{KS87}
H.~Karloff, and D.~Shmoys.
\newblock Efficient parallel algorithms for edge coloring problems.
\newblock {\em Journal of Algorithms}, 8(1), pages 39-52, 1987.

%\bibitem{K72}
%R.~Karp
%\newblock Reducibility among combinatorial problems.
%\newblock {\em Complexity of Computer Computations}, New York: Plenum Press, pages 85-103, 1972.

%\bibitem{K95}
%J.H.~Kim.
%\newblock The Ramsey number $R(3,t)$ has order of magnitude $t^2/ \log t$.
%\newblock {\em Random Structures and Algorithms}, 7:173-207, 1995.


%\bibitem{KW13}
%M.~K\"{o}nig,  and  R.~Wattenhofer.
%\newblock On local fixing.
%\newblock In {\em Proc. of the 17th International Conference on Principles of Distributed Systems}, pages 191 - 205, 2013. 


\bibitem{KSV11}
A.~Korman, J.~Sereni, and L.~Viennot.
\newblock Toward more localized local algorithms: removing assumptions concerning global knowledge.
\newblock In {\em Proc. of the 30th ACM Symp. on Principles of Distributed Computing}, pages 49-58, 2011.

%\bibitem{KP94}
%G.~Kortsarz, and D.~Peleg.
%\newblock Generating sparse 2-spanners.
%\newblock {\em Journal of Algorithms}, 17(2): 222-236, 1994.

%\bibitem{KP11}
%K.~Kothapalli, and S.~Pemmaraju.
%\newblock Distributed graph coloring in a few rounds.
%\newblock In {\em Proc. of the 30th ACM Symp. on Principles of Distributed Computing}, pages 31-40, 2011.

%\bibitem{KP12}
%K.~Kothapalli, and S.~Pemmaraju.
%\newblock Super-fast $3$-ruling sets.
%\newblock In {\em Proc. of the $32$nd IARCS International Conference on Foundations of Software Technology and Theoretical Computer Science}, pages 136 - 147, 2012.

%\bibitem{KSOS06}
%K.~Kothapalli, C.~Scheideler, M.~Onus, and C.~Schindelhauer.
%\newblock Distributed coloring in ~O($\sqrt{\log n}$) bit rounds.
%\newblock In {\em Proc. of the 20th International Parallel and Distributed Processing Symposium}, 2006.

\bibitem{K09}
F.~Kuhn.
\newblock Weak graph colorings: distributed algorithms and applications. 
\newblock In {\em Proc. of the 21st ACM Symposium on Parallel Algorithms and Architectures},  pages 138--144, 2009.

%\bibitem{KMNW05}
%F.~Kuhn, T.~Moscibroda, T.~Nieberg, and R.~Wattenhofer.
%\newblock Fast deterministic distributed maximal independent set computation on growth-bounded graphs.
%\newblock In {\em Proc. of the 19th International Symposium on Distributed Computing}, pages 273-287, 2005.

%\bibitem{KMW04}
%F.~Kuhn, T.~Moscibroda, and R.~Wattenhofer.
%\newblock What cannot be computed locally!
%\newblock In {\em Proc. of the 23rd ACM Symp. on Principles of Distributed Computing}, pages 300-309, 2004.

%\bibitem{KMW05}
%F.~Kuhn, T.~Moscibroda, and R.~Wattenhofer.
%\newblock On the locality of bounded growth.
%\newblock In {\em Proc. of the 24th ACM Symp. on Principles of Distributed Computing}, pages 60 -68, 2005.


%\bibitem{KMW10}
%F.~Kuhn, T.~Moscibroda, and R.~Wattenhofer.
%\newblock Local Computation: Lower and Upper Bounds.
%\newblock {\em http://ar{X}iv.org/abs/1011.5470}, 2010.

%\bibitem{KW05}
%F.~Kuhn and R.~Wattenhofer.
%\newblock Constant-time distributed dominating set approximation.
%\newblock {\em Distributed Computing}, 17(4): 303-310, 2005.

\bibitem{KW06}
F.~Kuhn, and R.~Wattenhofer.
\newblock On the complexity of distributed graph coloring.
\newblock In {\em Proc. of the 25th ACM Symp. on Principles of Distributed Computing}, pages {7--15}, 2006.

%\bibitem{LOW08}
%C.~Lenzen, Y.~Oswald, and R.~Wattenhofer.
%\newblock What Can Be Approximated Locally? Case Study: Dominating Sets in Planar Graphs.
%\newblock In {\em Proc 20th ACM Symp. on Parallelism in Algorithms and Architectures}, pages 46-54, 2008.
%\newblock See also {\em TIK report number 331}, ETH Zurich, 2010.

%\bibitem{LW10}
%C.~Lenzen, and R.~Wattenhofer.
%\newblock Minimum dominating set approximation in graphs of bounded arboricity.
%\newblock In {\em Proc. of the 24th Symp. on Distributed Computing}, pages 510 -524, 2010.

%\bibitem{LW11}
%C.~Lenzen and R.~Wattenhofer.
%\newblock MIS on trees.
%\newblock In {\em Proc. of the 30th ACM Symp. on Principles of Distributed Computing}, pages 41-48, 2011.

\bibitem{LSH96}
W.~Liang, X.~Shen, and  Q~.Hu.
\newblock Parallel Algorithms for the Edge-Coloring and Edge-Coloring Update Problems.
\newblock {\em Journal of Parallel and Distributed Computing}, 32(1), Pages 66-73, 1996.


\bibitem{L87}
N.~Linial.
\newblock Distributive graph algorithms: Global solutions from local data
\newblock In {\em Proc. of the 28th Annual Symp. on Foundation of Computer Science}, pages {331--335}, 1987.

%\bibitem{L92}
%N.~Linial.
%\newblock Locality in distributed graph algorithms.
%\newblock {\em SIAM Journal on Computing}, 21(1):193--201, 1992.

%\bibitem{LS92}
%N.~Linial and M.~Saks.
%\newblock Low diameter graph decomposition.
%\newblock {\em Combinatorica} 13: 441 - 454, 1993.

%\bibitem{L66}
%L.~Lov\'asz.
%\newblock On decompositions of graphs.
%\newblock {\em Studia Sci. Math. Hungar.}, 1:237--238, 1966.

%\bibitem{LPS88}
%A.~Lubotzky, R.~Phillips, and P.~Sarnak.
%\newblock Ramanujan graphs. 
%\newblock {\em Combinatorica}, 8(3): 261--277, 1988.

%\bibitem{L86}
%M.~Luby.
%\newblock A simple parallel algorithm for the maximal independent set problem.
%\newblock {\em SIAM Journal on Computing}, 15:1036-1053, 1986.

%\bibitem{L93}
%M.~Luby.
%\newblock Removing randomness in parallel computation without a processor penalty.
%\newblock In {\em Proc. of the 29th Annual Symposium on Foundations of Computer Science}, pages 162-173, 1988.

%\bibitem{L96}
%N.~Lynch.
%\newblock {\em Distributed Algorithms.}
%\newblock Morgan Kaufmann, 1996.
%\bibitem{MT01}

\bibitem{M55}
J.~Mycielski.
\newblock Sur le coloriage des graphes. 
\newblock {\em Colloq. Math.} 3, pages 161-162, 1955.

%B.~Mohar and C.~Thomassen.
%\newblock {\em Graphs on Surfaces.}
%\newblock Johns Hopkins University Press, 2001.

%\bibitem{MU05}
%M.~Mitzenmacher, and E.~Upfal.
%\newblock {\em Probability and Computing: Randomized Algorithms and Probabilistic Analysis.}
%\newblock Cambridge University Press, 2005.

\bibitem{MNN94}
R.~Motwani, J.~Naor, and M.~Naor.
\newblock The probabilistic method yields deterministic parallel algorithms. 
\newblock {\em Journal of Computer and System Sciences}, 49(3), 478-516, 1994.

%\bibitem{N14}
%D.~Nanongkai.
%\newblock Distributed approximation algorithms for weighted shortest paths.
%\newblock In {\em Proc. of the 46th ACM Symp. on Theory of Computing}, pages 565-573, 2014.

%\bibitem{M55}
%J.~Mycielski.
%\newblock Sur le coloriage des graphes.
%\newblock {\em Colloq. Math.} 3: 161-162, 1955.

%\bibitem{N91}
%M.~Naor.
%\newblock A lower bound on probabilistic algorithms for distributive ring coloring.
%\newblock {\em SIAM Journal on Discrete Mathematics}, 4(3):409-412, 1991.

%\bibitem{NS93}
%M.~Naor, and L.~Stockmeyer.
%\newblock What can be computed locally?
%\newblock In {\em Proc. 25th ACM Symp. on Theory of Computing}, pages 184-193, 1993.

%\bibitem{NW64}
%C.~Nash-Williams.
%\newblock Decompositions of finite graphs into forests.
%\newblock {\em J. London Math}, 39:12, 1964.



\bibitem{PR01}
A.~Panconesi, and R.~Rizzi.
\newblock Some simple distributed algorithms for sparse networks.
\newblock {\em Distributed Computing}, 14(2):97--100, 2001.

%\bibitem{PS95}
%A.~Panconesi, and A.~Srinivasan.
%\newblock On the complexity of distributed network decomposition.
%\newblock {\em Journal of Algorithms}, 20(2):581-Å92, 1995.

\bibitem{PS97}
A.~Panconesi, and A.~Srinivasan.
\newblock Randomized Distributed Edge Coloring via an Extension of the Chernoff-Hoeffding Bounds.
\newblock {\em SIAM Journal on Computing}, 26(2):350-368, 1997.

\bibitem{P00}
D.~Peleg.
\newblock {\em Distributed Computing: A Locality-Sensitive Approach.}
\newblock SIAM, 2000.

%\bibitem{SZ01}
%D.~Sanders, and Y.~Zhao.
%\newblock Planar Graphs of Maximum Degree Seven are Class I.
%\newblock {\em Journal of Combinatorial Theory, Series B}, 83(2):201?212, 2001.

%\bibitem{S06}
%N.~Santoro.
%\newblock {\em Design and Analysis of Distributed Algorithms.}
%\newblock Wiley, 2006.

%\bibitem{SV93}
%M.~Szegedy, and S.~Vishwanathan.
%\newblock Locality based graph coloring.
%\newblock In {\em Proc. of the 25th ACM Symp. on Theory of Computing}, pages 201-207, 1993.

\bibitem{V64}
V.~Vizing.
\newblock On an estimate of the chromatic class of a p-graph.
\newblock {\em Metody Diskret. Analiz}, 3: 25-30, 1964.

%\bibitem{V65}
%V.~Vizing.
%\newblock Critical graphs with given chromatic index.
%\newblock {\em Metody Diskret. Analiz}, 5: 9-17, 1965.

%\bibitem{SS12}
%R.~Saket, and M.~Sviridenko.
%\newblock New and improved bounds for the minimum set cover problem.
%\newblock {\em Approximation, Randomization, and Combinatorial Optimization. Algorithms and Techniques 
%Lecture Notes in Computer Science},  7408:288-300, 2012.

%\bibitem{SEW13}
%J.~Schneider, M.~Elkin, and R.~Wattenhofer.
%\newblock Symmetry breaking depending on the chromatic number or the neighborhood growth.
%\newblock {\em Theoretical Computer Science}, 509: 40-50, 2013.


%\bibitem{SW08}
%J.~Schneider, and R.~Wattenhofer.
%\newblock A log-star distributed Maximal Independent Set algorithm for Growth Bounded Graphs.
%\newblock In {\em Proc. of the 27th ACM Symp. on Principles of Distributed Computing}, pages 35--44, 2008.

%\bibitem{SW10}
%J.~Schneider, and R.~Wattenhofer.
%\newblock A new technique for distributed symmetry breaking.
%\newblock In {\em Proc. of the 29th ACM Symp. on Principles of Distributed Computing}, pages 257-266, 2010.

\bibitem{WHHHLSS97}
D.~Williamson, L.~Hall, J.~Hoogeveen, C.~Hurkens, J.~Lenstra, S.~Sevast'janov, and D.~ Shmoys.
\newblock Short shop schedules.
\newblock {\em Operations Research}, 45 (2), pages 288–-294, 1997.

\bibitem{ZN98}
X.~Zhou, and T.~Nishizeki.
\newblock Edge-coloring and f-coloring for various classes of graphs.
\newblock {\em Journal of Graph Algorithms and Applications}, 3(1), 1-18. 1998.

%\bibitem{SV93}
%M.~Szegedy, and S.~Vishwanathan.
%\newblock Locality based graph coloring.
%\newblock In {\em Proc. 25th ACM Symposium on Theory of Computing}, pages 201-207, 1993.

%\bibitem{Z07}
%D.~Zuckerman.
%\newblock Linear Degree Extractors and the Inapproximability of Max Clique and Chromatic Number.
%\newblock {\em Theory of Computing}, 3(1):103--128. 2007.


\end{thebibliography}
\end{document}